\documentclass[final]{IEEEtran}          
\usepackage{graphicx,multirow,booktabs,threeparttable,lscape}
\usepackage{amssymb,amsmath,amsbsy}
\usepackage{algorithm}
\usepackage{algorithmic}
\usepackage{graphicx}
\usepackage[caption=false,font=footnotesize]{subfig}
\usepackage{CJKutf8}
\usepackage{array}
\usepackage{multirow}
\usepackage{url}
\usepackage{verbatim} 
\newcommand\norm[1]{\left\lVert#1\right\rVert}

\def\thesection{\Roman{section}}

\newtheorem{lemma}{Lemma}

\begin{document}

\title{DNN-Assisted Near-Optimal Codebook Design for Uplink SCMA Schemes over Rician Fading Channels}
\author{Wasif Ali, {\it Student Member}, Yen-Ming Chen, {\it Member}, {\it IEEE}, and Fan-Shuo Tseng, {\it Senior Member}, {\it IEEE}
\thanks{This work is supported in part by National Science and Technology Council, Taiwan, under grants NSTC 115-2221-E-A49-081 and NSTC 115-2640-E-110-003-, and in part by the Sixth Generation Communication and Sensing Research Center funded by the Higher Education SPROUT Project, the Ministry of Education of Taiwan.
{\it (Corresponding Author: Yen-Ming~Chen).}}
\IEEEcompsocitemizethanks{\IEEEcompsocthanksitem Wasif Ali and Fan-Shuo Tseng are with the Institute of Communications Engineering, National Sun Yat-sen University, Kaohsiung 80424, Taiwan
	(e-mail: wasifali45nsyu@gmail.com; fs.tseng@mail.nsysu.edu.tw).}
\IEEEcompsocitemizethanks{\IEEEcompsocthanksitem Yen-Ming~Chen is with the Institute of Communications Engineering, National Yang Ming Chiao Tung University, Hsinchu 300093, Taiwan (e-mail: emerychen.cm@gmail.com).
}
}
\maketitle

\begin{abstract}
Sparse Code Multiple Access (SCMA) is a promising non-orthogonal multiplexing technique for enhancing connectivity and spectral efficiency in next-generation communication networks. This paper proposes near-optimal uplink SCMA codebook designs based on a novel deep neural network (DNN)-assisted framework. The DNN is trained in an unsupervised manner, with a loss function derived from the structural characteristics of the generated codebooks, thereby eliminating the need for labeled training data. Based on error-probability analysis, Rician-factor-dependent design criteria are derived for uplink Rician fading and used to formulate the loss function. The resulting criteria naturally encompass AWGN and Rayleigh fading as limiting cases, providing a unified framework for codebook design. Moreover, upper and lower bounds on the average mutual information (AMI) are derived, providing an information-theoretic interpretation of the PEP-based criterion and establishing its asymptotic consistency with the AMI-lower-bound-based criterion in the high-SNR regime. The analysis further reveals how the dominant design factors evolve as the channel transitions from Rayleigh fading to AWGN. Numerical results show that the proposed near-optimal codebooks achieve significant performance gains over representative benchmark schemes across a wide range of Rician factors and exhibit excellent error-rate performance in the error-floor region, making them particularly attractive for ultra-reliable applications.

\end{abstract}

\section{Introduction}
\IEEEPARstart{I}{n} recent years, non-orthogonal multiple access (NOMA) has attracted significant research attention due to its potential to improve spectral efficiency, support massive connectivity, and accommodate heterogeneous quality-of-service requirements, making it a promising candidate for beyond-5G and 6G systems \cite{6G}\cite{NOMA}. 
One representative example is power-domain NOMA (PD-NOMA) \cite{PD_NOMA}, where multiple user equipments (UEs) share the same resource elements (REs) by transmitting at different power levels and are separated at the receiver using successive interference cancellation (SIC).

In contrast, code-domain NOMA techniques multiplex the signals of different UEs through user-specific signatures or sparse code structures. One representative example is the low-density signature (LDS) scheme proposed in \cite{LDS_1} and \cite{LDS_2} for synchronous code-division multiple access (CDMA) systems. Among various code-domain NOMA techniques, sparse code multiple access (SCMA) \cite{SCMA1} has attracted particular interest because it integrates modulation and spreading into a unified framework. Specifically, each UE is assigned a carefully designed codebook, such that its information bits are directly mapped to sparse multidimensional codewords. This multidimensional mapping provides shaping gain and is regarded as one of the main factors contributing to the performance improvement over the repetition-based symbol mapping adopted in LDS \cite{LDS_1}.


A key factor affecting SCMA performance is codebook design, which has been extensively studied \cite{star}-\cite{VM}. In \cite{star}, the design is decomposed into factor-graph, rotation-angle, and mother codebook (MCB) optimization, with the inner-to-outer ring ratio of a star-QAM constellation adjusted to generate UE-specific signal points. In \cite{cpa}, codebooks are obtained by maximizing the constrained-input channel capacity using fixed pulse-amplitude modulation (PAM) constellations, whereas \cite{GA} directly optimizes the codebooks via a genetic algorithm without requiring a predefined mother constellation.

More recently, our previous work \cite{Proddis} developed high-performance SCMA codebooks and corresponding design criteria for several channel conditions, although direct optimization of the signal constellations was not fully addressed. For additive white Gaussian noise (AWGN) channels, \cite{SU} optimized random constellations by maximizing the  {\it minimum squared Euclidean distance} (MSED) among superimposed signal-vector pairs through biconvex optimization, while \cite{DE}, \cite{PI}, and \cite{LP} imposed specific mother-codeword structures to reduce the search space before applying various optimization algorithms. For uplink Rayleigh fading, \cite{RL-SCMA} employed reinforcement learning, treating signal-value selection as an action and the metric in \cite{Proddis} as the reward. These studies enable joint optimization of signal sub-constellations, but SCMA codebook design remains challenging for large-scale settings and practical fading environments beyond AWGN and Rayleigh channels. Moreover, most existing approaches rely on problem-specific optimization procedures, motivating a more flexible and systematic framework applicable to different fading environments.


Compared with the idealized AWGN channel and the purely non-line-of-sight (NLOS) Rayleigh fading channel, the Rician fading channel provides a more practical propagation model for many wireless environments by jointly capturing deterministic LOS and random multipath components. Such conditions commonly arise in short-range wireless links, relay-aided transmissions, unmanned aerial vehicle (UAV) communications, satellite communications, and various beyond-5G/6G applications. Therefore, the design of communication systems over Rician fading channels is of considerable practical importance. For SCMA, codebook design becomes more challenging because variations in the Rician factor change the relative strengths of the LOS and NLOS components, resulting in different channel characteristics that must be properly reflected in the codebook design.


Although several studies have considered SCMA codebook design for downlink Rician fading channels \cite{SCMA_NTN}, \cite{LP}, extension to the uplink case is nontrivial. In downlink transmission, the superimposed codeword on each RE is affected by a common channel coefficient, whereas in uplink SCMA, signals from different users experience independent fading coefficients. Under Rician fading, the coexistence of LOS and scattered components in these independently faded user signals further complicates the error analysis and the derivation of tractable design criteria. To the best of our knowledge, \cite{Tauf} is the only work that explicitly derives design criteria for uplink Rician SCMA. However, its design metrics do not explicitly depend on \(\kappa\), which may limit their effectiveness across different LOS conditions. Therefore, a unified design framework for uplink Rician SCMA over varying Rician factors remains to be established.



More broadly, recent advances in holographic multi-input multi-output (MIMO) have highlighted the increasingly rich spatial propagation characteristics expected in future 6G systems, where large and densely spaced antenna arrays may require unified near- and far-field channel representations as well as both angular and distance-dependent channel information \cite{Holographic_MIMO_1}\cite{Holographic_MIMO_2}. These developments further emphasize the importance of flexible channel-aware signal design under diverse propagation environments.

Benefiting from advances in big data, optimization algorithms, and computing resources, deep neural networks (DNNs) have been widely applied in various fields, including computer vision \cite{DNN_speach}, game playing \cite{DNN_game}, and speech processing \cite{DNN_computer}. DNNs have also been introduced into communication system design, with applications such as channel decoding \cite{DNN_decoder_1}, sparse signal recovery \cite{DNN_sparse_1}, end-to-end receiver design \cite{DNN_system_2}, and MIMO detection \cite{DNN_system_3}. In the context of SCMA system design, a multi-output DNN-based classification model was proposed in \cite{SCMA_DNN_receiver} for SCMA receivers, aiming to reduce detection complexity while maintaining satisfactory error-rate performance compared with conventional message passing algorithm (MPA) based solutions \cite{MPA} and \cite{MPA2}. On the other hand, end-to-end DNN-based SCMA design approaches have been investigated using denoising autoencoders (DAEs) \cite{deep_scma}-\cite{deep_scma_2}.
In these approaches, the encoding network is constructed to emulate the SCMA encoding process, where input bits are mapped to multidimensional codeword vectors and then superimposed across multiple users. The decoding network, in turn, models the receiver-side detection and decoding process. Random noise and channel coefficients are introduced between the encoding and decoding networks, and the loss function is defined based on end-to-end transmission simulations by comparing the original information bits with the decoded output bits.


Motivated by these developments, this work investigates the design of near-optimal SCMA codebooks over uplink Rician fading channels using a DNN-assisted framework. The main contributions of this paper are summarized as follows:

\begin{enumerate}
    \item {\it DNN-assisted Codebook Construction Framework:}
    A DNN-assisted framework is proposed for SCMA codebook construction under different fading environments. In contrast to end-to-end designs, the proposed network focuses only on transmitter-side codebook design and is trained in an unsupervised manner, with loss functions formulated from the structural characteristics of the generated codebooks. The codebook entries are directly parameterized by the DNN outputs and can be jointly optimized through gradient-based updates, providing a unified framework for multi-codebook design without requiring labeled training data.
    \item {\it Unified Design Criteria for Rician Fading Channels:}
    Codebook-design criteria are derived for uplink Rician fading channels with explicit dependence on the Rician factor. The resulting criteria naturally encompass Rayleigh and AWGN channels as limiting cases and reveal how the dominant design factors evolve with the relative strengths of the LOS and NLOS components.
    \item {\it AMI-Based Interpretation of the PEP-Based Design Criterion:}
    Upper and lower bounds on the average mutual information (AMI) are derived to provide an information-theoretic interpretation of the proposed PEP-based design criterion. High-SNR analysis further establishes the asymptotic consistency between the PEP-based criterion and that implied by the AMI lower bound.
    \item {\it Loss Function Formulation Across Different Rician Factors:}
    A log-domain loss formulation is developed to reliably account for different error events across Rician factors. Higher-order error events that become non-negligible in the strong-LOS regime are further incorporated through an extended loss formulation without exhaustive enumeration of all signal-matrix pairs.
    \item {\it Near-optimal SCMA Codebooks Across Various Channel Conditions:}
    The proposed framework generates near-optimal codebooks for AWGN, Rayleigh, and general Rician fading channels. Compared with representative benchmarks, the resulting designs achieve significant SNR gains over a wide range of Rician factors, particularly in the error-floor region, demonstrating their potential for ultra-reliable SCMA applications.
\end{enumerate}

To further clarify the distinction from our previous studies, \cite{Proddis} developed separate codebook-design criteria for AWGN and Rayleigh fading channels, whereas the present work derives a unified \(\kappa\)-dependent criterion for uplink Rician fading. Moreover, unlike the single-codebook action-based optimization adopted in \cite{RL-SCMA}, the proposed DNN-assisted framework jointly optimizes multiple UE codebooks via gradient-based updates. Our previous MIMO-SCMA study \cite{mimo_scma} primarily extended the methodology in \cite{Proddis} to MIMO-SCMA systems and additionally addressed receiver design, without considering the Rician-dependent design problem studied herein.


{\it Notation:} Matrices and vectors are denoted by boldface uppercase and lowercase letters, respectively.
$\mathbb{B}$, $\mathbb{C}$, and $\mathbb{R}$ denote the binary, complex, and real domains, respectively. $\mathcal{CN}(\mu, 2\sigma^2)$ denotes a circularly symmetric complex Gaussian random variable with mean $\mu$ and variance $2\sigma^2$.
The superscript $^{\top}$ denotes matrix transpose.
The notation $\norm{\cdot}$ represents the Frobenius norm, and $\displaystyle \mathop{\mathrm{E}}[\cdot]$
denotes the expectation
operator.

\section{Preliminaries}

In this paper, we focus on uplink SCMA systems comprising $K$ orthogonal REs and $L$ UEs, under the assumption that $K < L$. Each UE is allowed to access $d_\mathrm{v}$ REs, where $d_\mathrm{v} < K$, and each RE is shared by $d_\mathrm{f}$ UEs, where $d_\mathrm{f} < L$. Let $M$ denote the total number of codewords in an SCMA codebook. Accordingly, the SCMA encoding for the $l$-th UE is defined by the function
\begin{equation}
\begin{split}
\label{mapping_1}
 f_l:\mathbb{B}^{\log_2M\times1}\mapsto \mathcal{X}_l, \mbox{  i.e., }\mathbf{x}_l=f_l(\mathbf{u}_l),
\end{split}
\end{equation}
where $\mathbf{u}_l\in\mathbb{B}^{\log_2M\times1}$ denotes the
binary information vector for the $l$-th UE, and
$\mathbf{x}_l\in\mathbb{C}^{K\times1}$ denotes the multi-dimensional SCMA codeword
selected from codebook $\mathcal{X}_l$ for the $l$-th UE. If the power-balanced constraint is imposed, $\displaystyle \mathop{\mathrm{E}}[\parallel \mathbf{x}_l\parallel^2]=d_\mathrm{v}$.

In an uplink system, the $K \times 1$ received signal vector at the base station can be represented as
\begin{equation}
\begin{split}
\label{UP_received}
\mathbf{r} = \sqrt{\frac{E_{\mathrm{s}}}{d_{\mathrm{f}}}}\sum\limits^L_{l=1}\mathrm{diag}\{\mathbf{h}_l\}{\mathbf{x}_l}+\mathbf{n},
\end{split}
\end{equation}
where $\mathbf{h}_l\in \mathbb{C}^{K\times1}$ denotes the channel
vector from the $l$-th UE to the BS and $E_{\mathrm{s}}$ denotes the average energy of the received signal at a single RE.
Additionally, the vector $\mathbf{n} \in \mathbb{C}^{K \times 1}$ represents the additive noise at the base station, with each entry independently drawn from $\mathcal{CN}(0, N_\mathrm{0})$.

Since only $d_\mathrm{v}$ REs are available to the $l$-th UE, the
signal vector $\mathbf{x}_l$ is sparse, containing only $d_\mathrm{v}$ non-zero entries. 
Let $\mathbf{c}_l$ be a
$d_\mathrm{v}$-dimensional complex vector selected from the mother
codebook $\mathcal{C}_l\in \mathbb{C}^{d_\mathrm{v}\times1}$ following
\begin{equation}
\begin{split}
\label{mapping_2}
 g_l:\mathbb{B}^{\log_2M\times1}\mapsto \mathcal{C}_l, \mbox{  i.e., }\mathbf{c}_l=g_l(\mathbf{u}_l).
\end{split}
\end{equation}
Accordingly, we re-write the SCMA encoding process given in (\ref{mapping_1}) as
\begin{equation}
\begin{split}
\label{mapping_3}
f_l:\equiv \mathbf{V}_lg_l, \mbox{  i.e., }\mathbf{x}_l=\mathbf{V}_l\mathbf{c}_l,
\end{split}
\end{equation}
where the $K\times d_\mathrm{v}$ mapping matrix $\mathbf{V}_l$ is used to transfer a $d_\mathrm{v}$-dimensional
complex mother codeword to a $K$-dimensional
complex SCMA codeword, when considering the $l$-th
UE. It is worth noting that, unlike the definition in \cite{Proddis}, we omit the phase rotation factors originally included in $\mathbf{V}_l$, as these are implicitly handled within the DNN-assisted design process of $\mathcal{C}_l$. As a result, $\mathbf{V}_l$ solely functions to map the signals in $\mathbf{c}_l$ to the REs allocated to the $l$-th UE.
If we consider the case where $K=4$, $L=6$,
$d_\mathrm{v}=2$, and $d_\mathrm{f}=3$, and employ a $K\times L$
factor graph matrix
\begin{equation}\label{spase}
\begin{split}
\mathbf{F}=\left[
\begin{matrix}
1 & 1 & 1 & 0 & 0 & 0 \\
1 & 0 & 0 & 1 & 1 & 0 \\
0 & 1 & 0 & 1 & 0 & 1 \\
0 & 0 & 1 & 0 & 1 & 1 \\
\end{matrix}
\right],
\end{split}
\end{equation}
then $\mathbf{F}$
specifies the positions of the REs allocated to different
UEs. For example, the first column of $\mathbf{F}$ indicates that UE 1 is allowed to access the second and fourth REs. Consequently, the matrix $\mathbf{V}_l$ for UE $l$ can be derived from the $l$-th column of $\mathbf{F}$. Specifically, for UE 1 and UE 2, we have
\begin{equation}\label{V_matrix}
\mathbf{V}_1=\left[
\begin{matrix}
1 & 0\\
0 & 1\\
0 & 0\\
0 & 0\\
\end{matrix}\right]\mbox{ and }\mathbf{V}_2=\left[
\begin{matrix}
1 & 0\\
0 & 0\\
0 & 1\\
0 & 0\\
\end{matrix}\right].
\end{equation}
Lastly, for the receivers, we consider the conventional MPA receiver reported in \cite{MPA} and \cite{MPA2}.

\section{DNN-assisted codebook Construction for Near-optimal SCMA Schemes}

Considering the previously proposed SCMA codebook construction methods \cite{Proddis}, although near-optimal performance can potentially be achieved, the design complexity increases rapidly with the number of users and modulation levels. Consequently, these methods are more suitable for small-scale SCMA systems. In contrast, recent advancements in artificial intelligence and computational resources have enabled the use of neural networks to approximate complex optimization problems and approach near-optimal solutions via gradient-based methods. In this paper, we propose a flexible DNN-assisted SCMA codebook design approach that aims to achieve near-optimal error performance with reduced design complexity under different fading environments.



\begin{figure*}[!t]
	\centering
	\includegraphics[width = 7.2in]{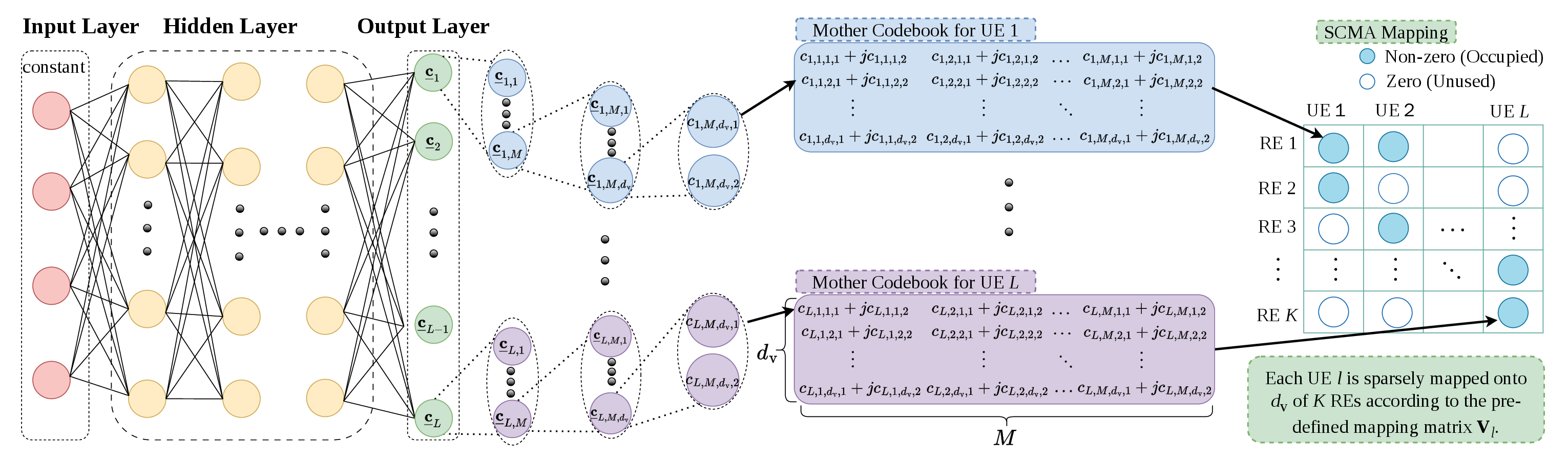}
	\caption{Overview of the proposed DNN-assisted SCMA codebook generation framework.}
	\label{dnn_ul}
\end{figure*}

\subsection{Proposed DNN-assisted Design Approach for the SCMA Codebook Design}
A simple DNN-based architecture
is adopted for SCMA codebook construction. Consider a basic DNN framework consisting of $Z-1$ fully connected hidden layers. The $z$-th hidden layer is characterized by a weight matrix $\mathbf{W}_{\mathrm{h},z}$, a bias vector $\mathbf{b}_{\mathrm{h},z}$, and a nonlinear activation function $f(\cdot)$, where $z=1,2,\ldots,Z-1$.
A rectified linear unit (ReLU), denoted by $f(a) = \max(0, a)$ for an arbitrary input $a$, serves as the activation function. Consequently, the desired output neuron sequence $\underline{\mathbf{c}}$ generated from a DNN can be expressed as
\begin{equation}
	\underline{\mathbf{c}}=\mathbf{W}_{\mathrm{out}} f(\cdots f(\mathbf{W}_{\mathrm{h},2} f (\mathbf{W}_{\mathrm{h},1}\mathbf{a}+\mathbf{b}_{\mathrm{h},1})+\mathbf{b}_{\mathrm{h},2})\cdots) + \mathbf{b}_{\mathrm{out}},
\end{equation}
where $\mathbf{a}$ denotes the $I_\mathrm{a} \times 1$ input vector and $\underline{\mathbf{c}}$ denotes the $I_{\mathrm{out}} \times 1$ output vector. We further assume that $I_{\mathrm{h},z}$ nodes are contained in each hidden layer. Accordingly, we have $\mathbf{W}_{\mathrm{h},1}\in \mathbb{R}^{I_{\mathrm{h},1}\times I_\mathrm{a}}$, $\mathbf{b}_{\mathrm{h},1}\in \mathbb{R}^{I_{\mathrm{h},1}\times 1}$, $\mathbf{W}_{\mathrm{h},z}\in \mathbb{R}^{I_{\mathrm{h},z}\times I_{\mathrm{h},z-1}}$ and $\mathbf{b}_{\mathrm{h},z}\in \mathbb{R}^{I_{\mathrm{h},z}\times 1}$, $z=2,\cdots,Z-1$. For the output layer, we assume $\mathbf{W}_{\mathrm{out}}\in \mathbb{R}^{I_{\mathrm{out}}\times I_{\mathrm{h},Z-1}}$ and $\mathbf{b}_{\mathrm{out}}\in \mathbb{R}^{I_{\mathrm{out}}\times 1}$. The values of $I_\mathrm{a}$ and $I_{\mathrm{h},z}$ are determined manually, while the value of $I_{\mathrm{out}}$ depends on the sizes of codebooks to be constructed. In the implementation, we adopt a progressively contracting topology with hidden-layer widths given by $(32,16,8,4,2)I_{\mathrm{out}}$. These widths are implementation hyperparameters and are not specified by the SCMA formulation.

%

\begin{table}[!t]
	\centering
	\caption{DNN architecture parameters and optimization configuration}
	\label{tab_dnn_settings}
	\renewcommand{\arraystretch}{1.2}
	\setlength{\tabcolsep}{6.5pt}
	\begin{tabular}{ll}
		\hline
		\textbf{Parameter} & \textbf{Setting} \\
		\hline
		Network type & Fully connected DNN \\
		Input & Fixed $\mathbf{a}=\mathbf{1}\in\mathbb{R}^{4}$ \\
		Output size & $I_{\mathrm{out}}=2Jd_{\mathrm{v}}M$ \\
		Number of Hidden Layers & 5\\
		Width of Hidden layers & $(32,16,8,4,2)I_{\mathrm{out}}$ \\
		Activations & ReLU (hidden layers); Linear (output layer) \\
		Initialization & Kaiming uniform \\
		Optimizer & Adam \\
		Learning rate & $\eta=10^{-3}$ \\
		Batch setting & Single fixed input \\
		Objective & Proposed PEP-based loss \\
		Parameter update & Backpropagation through DNN \\
		Training data & Not required \\
		Model selection & Minimum-loss parameters \\
		Termination & Convergence or $E_{\max}$ iterations \\
		\hline
	\end{tabular}
\end{table}

In a conventional DNN architecture, the loss function is typically used to quantify the difference between the network outputs and the corresponding target values. Therefore, the availability of reference labels or predefined ``correct answers" is usually required for supervised training.
Existing DAE-based design methodologies \cite{deep_scma}-\cite{deep_scma_2} provide an effective end-to-end learning framework for communication system design. Nevertheless, their direct application to SCMA codebook optimization may involve several practical challenges. In particular, these methods usually require joint optimization of both the transmitter and receiver, resulting in a relatively large neural network and increased training complexity. Moreover, when the actual bit error ratio (BER) is adopted as the loss function, extensive Monte Carlo simulations are typically required at each backpropagation step, which can substantially increase the training time. This issue becomes more pronounced for large-scale SCMA systems, where an efficient design procedure is highly desirable.


Given the effectiveness of conventional MPA detectors for sparse codebooks, this work retains the MPA detector at the receiver and focuses instead on DNN-based SCMA codebook design. Accordingly, the proposed DNN is used solely as a parameterized SCMA codebook generator, rather than an end-to-end neural communication system. The architecture of the proposed DNN-assisted framework is shown in Fig.~\ref{dnn_ul}. Unlike conventional DNN architectures, the input vector $\mathbf{a}$ is set as a fixed arbitrary constant and remains unchanged throughout training. As detailed in the following subsections, the proposed loss functions are formulated based on analytical error-performance analysis, rather than directly defined in terms of the actual BER, and do not require reference labels. Therefore, $\mathbf{a}$ serves only as an initialization input to the network. Experimental results further show that the specific choice of $\mathbf{a}$ has negligible impact on both the convergence behavior and the resulting codebook quality. Therefore, an all-one vector is adopted for $\mathbf{a}$.


After passing $\mathbf{a}$ through $Z-1$ hidden layers, the output vector $\underline{\mathbf{c}}$ is obtained and interpreted as the entries in the SCMA MCBs. Since neurons can only represent real-valued quantities, each non-zero complex codebook entry is represented by two neurons. Hence, for an SCMA system with $L$ users, a total of $I_\mathrm{c}=2LMd_\mathrm{v}$ output neurons are required. For simplicity, we denote $c_{l,m,d,k}$, where $l=1,2,\ldots,L$, $m=1,2,\ldots,M$, $d=1,2,\ldots,d_\mathrm{v}$, and $k=1,2$, as a single output neuron in $\underline{\mathbf{c}}$ which represents the $k$-th real-valued component corresponding to the $d$-th entry of the $m$-th mother codeword for the $l$-th user. We further define the $m$-th mother codeword in $\mathcal{C}_l$ as
$\mathbf{c}_{l,m} \triangleq [ ( c_{l,m,1,1} + j c_{l,m,1,2} ), ( c_{l,m,2,1} + j c_{l,m,2,2} ), \ldots, ( c_{l,m,d_\mathrm{v},1} + j c_{l,m,d_\mathrm{v},2} ) ]^{\top}$.
If the power-balanced constraint is imposed, the expected energy of each codeword satisfies
$\mathop{\mathrm{E}}_m [ | \mathbf{c}_{l,m} |^2 ] = d_\mathrm{v}$ for all $l$.

During optimization, the loss gradient is backpropagated through the codebook-generation network, and all weight matrices and bias vectors are jointly updated using the Adam optimizer. The parameter set that yields the minimum loss is retained, and the procedure terminates when the prescribed convergence criterion is satisfied or the maximum number of iterations is reached. The adopted DNN architecture and optimization configuration are summarized in Table~I. With this general optimization procedure established, the following subsections develop the corresponding loss functions for different uplink channel environments, including AWGN, Rayleigh fading, and Rician fading channels.

\subsection{Rician Fading Channels}

The Rician fading channel is a widely used statistical model for multipath propagation, in which a dominant LOS component coexists with weaker scattered and reflected NLOS components. Accordingly, the channel coefficient can be modeled as the sum of a deterministic LOS term and a random NLOS term, given by
\begin{equation}
	\begin{split}
		\label{Rician-channel}
		h = \sqrt{\frac{\kappa}{1+\kappa}} h_{\mathrm{LOS}}+ \sqrt{\frac{1}{1+\kappa}} h_{\mathrm{NLOS}},
	\end{split}
\end{equation}
where $\kappa$ denotes the Rician $K$-factor that quantifies the power ratio between the dominant LOS component and the scattered NLOS components. The LOS component is given by
$h_{\mathrm{LOS}} = \exp( j 2\pi f_{\mathrm{d}} \cos \theta_0 + j \phi_0 )$, where $f_{\mathrm{d}}$ is the maximum Doppler shift, $\theta_0$ is the angle of arrival, and $\phi_0$ is an initial phase uniformly distributed over $[-\pi,\pi)$. In contrast, $h_{\mathrm{NLOS}}$ is modeled as a zero-mean, unit-variance complex Gaussian random variable. Equivalently, the channel coefficient $h$ can be modeled as an independent complex Gaussian random variable with distribution $\mathcal{CN}(\mu, 2\sigma^2)$, where $\kappa = \frac{\mu^2}{2\sigma^2}$.



In the absence of an LOS component, i.e., when $\kappa = 0$, the channel reduces to Rayleigh fading, where propagation is entirely governed by scattered and reflected paths. As $\kappa$ increases, the LOS component becomes increasingly dominant, and the channel behavior gradually approaches that of an AWGN channel, with reduced fading severity \cite{Proakis}.


\subsection{Derivation of Critical Distances}
For uplink Rician fading channels, if a single RE is considered and the signals from different UEs are viewed as originating from multiple transmit antennas, the SCMA scheme can be modeled as a multiple-input single-output (MISO) system with $d_\mathrm{f}$ transmit antennas and one receive antenna \cite{Proddis}. Moreover, by interpreting different REs as distinct time instants, the SCMA signal matrix design becomes analogous to the construction of a $d_\mathrm{f}\times K$ space-time code for fast fading channels. This interpretation is consistent with the principles established in prior works on space-time coding \cite{tri1}, \cite{STTC_Vucetic}.

We define $\mathbf{s}_{k} = [x_{l_{k,1},k}, x_{l_{k,2},k}, \ldots, x_{l_{k,d_{\mathrm{f}}},k}]^{\top}$ as the transmit signal vector associated with the $k$-th RE, where $l_{k,i}$, for $i = 1, 2, \ldots, d_{\mathrm{f}}$, denotes the indices of the UEs accessing the $k$-th RE, and $x_{l_{k,i},k}$ represents the symbol transmitted by the UE indexed by $l_{k,i}$ over this RE. Let $r_k$ denote the $k$-th entry of the received signal vector $\mathbf{r}$. Accordingly, the received signal in (\ref{UP_received}) can be rewritten as
\begin{equation}
	\label{up_Riciank}
	\begin{split}
	r_k = \sqrt{\frac{E_{\mathrm{s}}}{d_{\mathrm{f}}}} \sum\limits_{i=1}^{d_{\mathrm{f}}} h_{l_{k,i},k} x_{l_{k,i},k} + n_k=\sqrt{\frac{E_{\mathrm{s}}}{d_{\mathrm{f}}}} \mathbf{h}_k^{\top}\mathbf{s}_k + n_k,
	\end{split}
\end{equation}
where $h_{l_{k,i},k}$ represents the channel coefficient between the $l_{k,i}$-th UE, which accesses the $k$-th RE, and the base station. The vector $\mathbf{h}_k\in \mathbb{C}^{d_\mathrm{f}}$ is defined as  $\mathbf{h}_k\equiv[h_{l_{k,1},k},h_{l_{k,2},k},\ldots,h_{l_{k,d_\mathrm{f}},k}]^{\top}$. 
Since the indices $l_{k,1}, l_{k,2}, \ldots, l_{k,d_{\mathrm{f}}}$ vary across different REs, the corresponding channel coefficients $h_{l_{k,i},k}$ depend on both $l_{k,i}$ and $k$, thereby modeling the channel as a fast Rician fading MISO channel.
We further define $\mathbf{S}=[\mathbf{s}_1,\mathbf{s}_2,\ldots,\mathbf{s}_K]$ as the transmitted signal matrix over $K$ REs. Accordingly, the erroneously decoded signal matrix is denoted as $\mathbf{\hat{S}} = [\mathbf{\hat{s}}_1, \mathbf{\hat{s}}_2, \ldots, \mathbf{\hat{s}}_K]$.
Following the analytical framework presented in \cite{STTC_Vucetic}, the derivation of the pair-wise error probability (PEP) between two candidate transmit signal matrices, $\mathbf{S}$ and $\mathbf{\hat{S}}$, entails evaluating the probability that the maximum likelihood (ML) decision metric for $\mathbf{\hat{S}}$ exceeds that of $\mathbf{S}$.
This relationship can be formally expressed as
\begin{equation}
	\begin{split}
		\label{up_RicianQ}
		\begin{aligned}
			&P(\mathbf{S} \rightarrow \mathbf{\hat{S}}|\mathbf{h}_1,\mathbf{h}_2,\cdots,\mathbf{h}_K)\\
			& = Q \left(\sqrt{\frac{E_{\mathrm{s}}}{2d_{\mathrm{f}}N_{0}} \sum\limits_{k\in\rho(\mathbf{S},\mathbf{\hat{S}})}\parallel \mathbf{h}_k^{\top}(\mathbf{s}_{k}-{\mathbf{\hat{s}}_{k}}) \parallel^2} \right),
		\end{aligned}
	\end{split}
\end{equation}
where ${\rho(\mathbf{S},\mathbf{\hat{S}})}$ identifies the set of RE indices for which $\mathbf{s}_k \neq \mathbf{\hat{s}}_k$ and $Q(x) \equiv \frac{1}{\sqrt{2\pi}}\int_{x}^{\infty} \exp\left(\frac{-t^{2}}{2}\right)dt$.
By applying the Chernoff bound,
the PEP can be conservatively estimated as
\begin{equation}
	\begin{split}
		\label{up_RicianPeP}
		\begin{aligned}
			&P(\mathbf{S} \rightarrow \mathbf{\hat{S}}|\mathbf{h}_1,\mathbf{h}_2,\cdots,\mathbf{h}_K)\\
			& \leq\ \frac{1}{2} \exp\left(\frac{-E_{\mathrm{s}}}{4d_{\mathrm{f}}N_{0}}\sum\limits_{k\in\rho(\mathbf{S},\mathbf{\hat{S}})}\parallel \mathbf{h}_k^{\top}(\mathbf{s}_{k}-{\mathbf{\hat{s}}_{k}}) \parallel^{2}\right).
		\end{aligned}
	\end{split}
\end{equation}

This formulation provides a rigorous assessment of the transmission reliability under fading conditions. 
Specifically, the faded squared Euclidean distance is quantified as follows:

\begin{equation}
	\begin{split}
		\label{up_RicianEU}
		\begin{aligned}
			& \sum\limits_{k\in\rho(\mathbf{S},\mathbf{\hat{S}})} \parallel \mathbf{h}_k^{\top}(\mathbf{s}_{k}-{\mathbf{\hat{s}}_{k}}) \parallel^{2} \\
			& = \sum\limits_{k\in\rho(\mathbf{S},\mathbf{\hat{S}})} \left\vert \sum\limits^{d_\mathbf{f}}_{i=1} h_{l_{k,i},k}{\Delta{x}_{l_{k,i},k}} \right\vert^{2}  = \sum\limits_{k\in\rho(\mathbf{S},\mathbf{\hat{S}})} \mid \Delta_k \mid^{2} ,
		\end{aligned}
	\end{split}
\end{equation}
where $|\Delta_k|^2$ denotes the faded squared Euclidean distance on the $k$-th RE.
To evaluate the average effect of the channel on this distance, we average with respect to the channel coefficients and write
\begin{equation}
	\begin{split}
		\label{up_RicianE}
		\begin{aligned}
			& P(\mathbf{S} \rightarrow \mathbf{\hat{S}}) \leq \frac{1}{2} \displaystyle \mathop{\mathrm{E}} \left[\exp\left(\frac{-E_{\mathrm{s}}}{4d_{\mathrm{f}}N_{0}}\sum\limits_{k\in\rho(\mathbf{S},\mathbf{\hat{S}})} \mid \Delta_k \mid^{2}\right) \right] \\
			& = \frac{1}{2} \prod\limits_{k\in\rho(\mathbf{S},\mathbf{\hat{S}})} \displaystyle \mathop{\mathrm{E}} \left[\exp\left(\frac{-E_{\mathrm{s}}}{4d_{\mathrm{f}}N_{0}} \mid \Delta_k \mid^{2}\right) \right].
		\end{aligned}
	\end{split}
\end{equation}

The above derivation assumes statistically independent small-scale fading coefficients across distinct UE-RE pairs. While the coefficients associated with different UEs generally correspond to different wireless links, those associated with different REs of the same UE may exhibit correlation, depending on the physical realization and mapping of the REs. Nevertheless, the independence assumption can serve as a reasonable approximation when the SCMA REs are mapped or interleaved over physical resources with sufficiently low channel correlation. In an orthogonal frequency-division multiple access (OFDMA)-based realization, this can be achieved by mapping different SCMA REs onto sufficiently separated subcarriers.

Since $h_{l_{k,i},k}$ are modeled as independent non-central circularly symmetric complex Gaussian random variables, their linear combination $\Delta_k$ also follows a non-central circularly symmetric complex Gaussian distribution, i.e., $\Delta_k\sim\mathcal{CN}(\mu_k,2\sigma_k^2)$. Consequently, the envelope \(|\Delta_k|\) follows a Rician distribution with noncentrality parameter \(|\mu_k|\) and scale parameter \(\sigma_k\), whose probability density function (PDF) is given by \cite{Proakis}
\begin{equation}
	\begin{split}
		\label{up_RicianB}
		\begin{aligned}
			P_{\mid \Delta_k\mid} (r) = \frac{r}{\sigma^2_k} \exp \left(-\frac{r^2 + |\mu_k|^2}{2 \sigma^2_k}\right)I_0\left(\frac{r\mid\mu_k\mid}{\sigma^2_k} \right),
		\end{aligned}
	\end{split}
\end{equation}
where $I_0(\mathbf{.})$ denotes the zeroth-order modified Bessel function of the first kind. Following (\ref{up_RicianEU}), the parameters $|\mu_k|^2$ and $\sigma^2_k$ are given by
\begin{equation}
	\begin{split}
		\label{up_RicianMV}
		\begin{aligned}
			|\mu_k|^2 = \left\vert \mu \sum\limits^{d_\mathbf{f}}_{i=1}{\Delta{x}_{{l_{k,i}},k}}  \right\vert^2, \ \sigma^2_k = \sigma^2\sum\limits^{d_\mathbf{f}}_{i=1}{\vert\Delta{x}_{l_{k,i},k}\vert^2}.
		\end{aligned}
	\end{split}
\end{equation}
When the effective Rician factor associated with $\Delta_k$, given by $|\mu_k|^2/(2\sigma^2_k)$, is sufficiently large, the Rician envelope $|\Delta_k|$ approaches a Gaussian random variable with mean $|\mu_k|$ and variance $\sigma_k^2$. Therefore, in the strong-LOS regime, $|\mu_k|$ and $\sigma_k^2$ can be approximately interpreted as the mean and variance of $|\Delta_k|$, respectively.

To account for the randomness of $\Delta_k$ and to provide a robust performance measure that captures the variability of the channel, the PDF described using (\ref{up_RicianB}) and (\ref{up_RicianMV}) is integrated with (\ref{up_RicianE}) to derive the unconditional PEP, where
\begin{equation}
	\begin{split}
		\label{up_RicianUPEP_1}
		\begin{aligned}
			P(\mathbf{S} \rightarrow \mathbf{\hat{S}}) \leq \frac{1}{2} \prod\limits_{k\in\rho(\mathbf{S},\mathbf{\hat{S}})} \frac{1}{1+\frac{\sigma^2_k E_{\mathrm{s}}}{2 d_{\mathrm{f}} N_\mathrm{0}}} \exp\left (- \frac{|\mu_k|^2\frac{E_{\mathrm{s}}}{4d_{\mathrm{f}}N_\mathrm{0}}}{1+\frac{\sigma^2_k E_{\mathrm{s}}}{2d_{\mathrm{f}}N_\mathrm{0}}} \right).
		\end{aligned}
	\end{split}
\end{equation}
Under the high-SNR assumption, i.e., $\frac{\sigma_k^2E_\mathrm{s}}{(2d_\mathrm{f}N_0)}\gg1$ for $k\in\rho(\mathbf S,\hat{\mathbf S})$, the unity term in each denominator of (\ref{up_RicianUPEP_1}) can be neglected. Accordingly, the PEP can be approximated as
\begin{equation}
	\label{up_RicianUB_1}
	\begin{aligned}
		& P(\mathbf{S} \rightarrow \mathbf{\hat{S}}) \\
		& \leq \frac{1}{2} \left( {\frac{E_{\mathrm{s}}}{2d_{\mathrm{f}} N_\mathrm{0}}}\right)^{-\delta} \prod\limits_{k\in\rho(\mathbf{S},\mathbf{\hat{S}})} \left( (\sigma^2_k)^{-1} \exp\left (- \frac{|\mu_k|^2}{2\sigma^2_k} \right)\right)\\
		& =  \left( {\frac{E_{\mathrm{s}}}{2d_{\mathrm{f}} N_\mathrm{0}}}\right)^{-\delta} \frac{1}{d^2_{\mathrm{c}}(\mathbf{S},\mathbf{\hat{S}})},
	\end{aligned}
\end{equation}
where $\delta \equiv \left| \rho(\mathbf{S}, \hat{\mathbf{S}}) \right|$ denotes the number of RE indices (i.e., columns) at which the transmitted and erroneously decoded signal matrices differ, and is referred to as the {\it symbol distance} (SD). In addition, the term $d^2_{\mathrm{c}}$ is defined as the {\it critical distance} (CD) between the signal-matrix pair $(\mathbf{S}, \hat{\mathbf{S}})$. The CD, which captures the combined impact of channel variations and signal discrepancies, is formally defined as
\begin{equation}
	\begin{split}
		\label{CD_1}
		\begin{aligned}
			d^2_{\mathrm{c}}(\mathbf{S},\mathbf{\hat{S}}) = 2\prod\limits_{k\in\rho(\mathbf{S},\mathbf{\hat{S}})} \left(\sigma_k^2 \exp\left(\frac{|\mu_k|^2}{2\sigma_k^2}\right)\right).
		\end{aligned}
	\end{split}
\end{equation}

In (\ref{up_RicianUB_1}), the SD value \(\delta\) determines the diversity order, whereas \(d_\mathrm{c}^2(\mathbf S,\hat{\mathbf S})\) characterizes the corresponding effective-distance or coding-gain contribution. For signal-matrix pairs with the same SD, a larger \(d_\mathrm{c}^2(\mathbf S,\hat{\mathbf S})\) corresponds to a smaller asymptotic PEP and therefore a larger coding gain. Following the design criteria for space-time coding as outlined in \cite{tri1}, \cite{STTC_Vucetic}, a fundamental objective is to maximize the minimum SD (MSD) among all possible $(\mathbf{S},\mathbf{\hat{S}})$ pairs. This maximization is crucial for optimizing the error performance, since the value of $\delta$ directly determines the slope of the error rate curve.

In the context of an SCMA scheme, the MSD value is inherently fixed at $d_\mathrm{v}$ \cite{Proddis}, due to the underlying system architecture in which the signals of each UE are mapped onto a predetermined set of $d_\mathrm{v}$ REs. 
Given this setting, a key design objective is to optimize the design metric derived from the CD values of signal-matrix pairs, under the given MSD. For the calculation of CD values, we consider the following two cases:


{\it Case 1: The special case where $\delta=d_\mathrm{v}$}

By properly designing the factor graph matrix such that each column has weight $d_\mathrm{v}$ and each row has weight $d_\mathrm{f}$, this condition holds if and only if one UE transmits different codewords in $\mathbf{S}$ and $\mathbf{\hat{S}}$, while all other UEs transmit the same codewords in both matrices. In this case, $\rho(\mathbf{S},\mathbf{\hat{S}})$ contains exactly $d_\mathrm{v}$ RE indices, and for each $k \in \rho(\mathbf{S},\mathbf{\hat{S}})$, the vectors $\mathbf{s}_{k}$ and $\hat{\mathbf{s}}_{k}$ differ in only one element. Suppose that, among the MCBs $\mathcal{C}_{1}, \mathcal{C}_{2}, \ldots, \mathcal{C}_{L}$ generated by the DNN, the $l$-th UE transmits the $p$-th and $q$-th mother codewords in $\mathbf{S}$ and $\mathbf{\hat{S}}$, respectively. Based on (\ref{up_RicianMV}) and (\ref{CD_1}), the CD corresponding to the signal-matrix pair with $\delta = d_\mathrm{v}$ is given by
\begin{equation}\label{CD_2}
	\begin{aligned}
	d^2_{\mathrm{c}}(\mathbf{S},\mathbf{\hat{S}}) &= 2\sigma^{2d_\mathrm{v}} \exp\left(\kappa d_\mathrm{v}\right) \prod\limits_{k\in\rho(\mathbf{S},\mathbf{\hat{S}})} \mid {c}_{l,p,k} - {c}_{l,q,k}\mid^{2}\\
	&\equiv[\gamma^{(d_\mathrm{v})}_{\mathrm{c},n}]^{-1},
	\end{aligned}
\end{equation}
where $n$ is indexed by $l$, $p$, and $q$. For each UE $l$ with codebook size $M$, there are $U=\tbinom{M}{2}$ corresponding values of $\gamma^{(d_\mathrm{v})}_{\mathrm{c},n}$. Therefore, for the case $\delta = d_\mathrm{v}$, the total number of signal-matrix pairs is $LU$, with the same number of corresponding $\gamma^{(d_\mathrm{v})}_{\mathrm{c},n}$ values. Thus, $n$ ranges from $1$ to $LU$.


\begin{figure*}
	\begin{equation}\label{CD_3}
		\begin{aligned}
			d^2_{\mathrm{c}}(\mathbf{S},\mathbf{\hat{S}})&=2\sigma^{2\delta}\exp\left(\kappa\sum\limits_{k\in\rho(\mathbf{S},\mathbf{\hat{S}})} \frac{\mid \sum_{t\in \varrho(\mathbf{s}_{k},\hat{\mathbf{s}}_{k})}(s_{t,k}-\hat{s}_{t,k})\mid^2}{\sum_{t\in \varrho(\mathbf{s}_{k},\hat{\mathbf{s}}_{k})}\left|s_{t,k}-\hat{s}_{t,k}\right|^{2}}\right)\prod_{k\in\rho(\mathbf{S},\mathbf{\hat{S}})}\left[\sum_{t\in \varrho(\mathbf{s}_{k},\hat{\mathbf{s}}_{k})}\left|s_{t,k}-\hat{s}_{t,k}\right|^{2}\right]\equiv [\gamma^{(\delta)}_{\mathrm{c},n}]^{-1} \\
		\end{aligned}
	\end{equation}
	\hrulefill
\end{figure*}

{\it Case 2: The cases where $\delta\neq d_\mathrm{v}$}

For the convenience of presentation, we re-define $s_{t,k}$ as the $t$-th entry of $\mathbf{s}_{k}$, and let $\varrho(\mathbf{s}_{k},\hat{\mathbf{s}}_{k})$ represent the set of entries (UEs) where $s_{t,k}\neq \hat{s}_{t,k}$. Accordingly, substituting (\ref{up_RicianMV}) into (\ref{CD_1}), using $\kappa=\frac{\mu^2}{2\sigma^2}$, and collecting the corresponding product and exponential terms yields the CD expression in (\ref{CD_3}). The parameter $n$, which indexes the CD value, depends on the selection pattern of $\delta$ REs, the indices of the UEs exclusively associated with these REs, and the corresponding transmitted codeword indices.

The exponent term in (\ref{CD_3}), excluding the factor \(\kappa\), is referred to as the {\it weighted squared Euclidean distance} (WSED) between \((\mathbf S,\hat{\mathbf S})\). For each \(k\in\rho(\mathbf S,\hat{\mathbf S})\), let $N_k \triangleq \mid\varrho(\mathbf{s}_{k},\hat{\mathbf{s}}_{k})\mid$ and define the vector of nonzero symbol differences as $\mathbf{\epsilon}_k\equiv[s_{t,k}-\hat{s}_{t,k}]_{t\in \varrho(\mathbf{s}_{k},\hat{\mathbf{s}}_{k})}\in\mathbb{C}^{N_k}$. By applying the Cauchy-Schwarz inequality, we have
\begin{equation}
	\mid\mathbf{1}_{N_k}^{H}\mathbf{\epsilon}_k\mid^2\leq \parallel\mathbf{1}_{N_k}\parallel^2 \parallel\mathbf{\epsilon}_k\parallel^2,
\end{equation}
which gives
\begin{equation}
	\frac{
		\left|
		\sum_{t\in\varrho(\mathbf s_k,\hat{\mathbf s}_k)}
		\left(s_{t,k}-\hat{s}_{t,k}\right)
		\right|^2
	}{
		\sum_{t\in\varrho(\mathbf s_k,\hat{\mathbf s}_k)}
		\left|s_{t,k}-\hat{s}_{t,k}\right|^2
	}
	\leq
	|\varrho(\mathbf s_k,\hat{\mathbf s}_k)|
	\label{resp_cauchy_expanded}
\end{equation}
for each \(k\in\rho(\mathbf S,\hat{\mathbf S})\). Therefore, the WSED is upper-bounded by $\sum_{k\in\rho(\mathbf{S},\hat{\mathbf{S}})}|\varrho(\mathbf s_k,\hat{\mathbf s}_k)|$. For a given RE, equality holds when all nonzero symbol differences \(s_{t,k}-\hat{s}_{t,k}\), \(t\in\varrho(\mathbf{s}_{k},\hat{\mathbf{s}}_{k})\), are identical in both magnitude and phase. Consequently, the exponent in (\ref{CD_3}) is upper-bounded by $\kappa\sum_{k\in\rho(\mathbf{S},\hat{\mathbf{S}})}|\varrho(\mathbf s_k,\hat{\mathbf s}_k)|$. In the special case of \(\delta=d_\mathrm{v}\), only one UE differs on each involved RE, such that \(|\varrho(\mathbf{s}_{k},\hat{\mathbf{s}}_{k})|=1\). Hence, the above bound is attained trivially and the exponent reduces to \(\kappa d_\mathrm{v}\), consistent with (\ref{CD_2}).



\subsection{AMI-Based Interpretation of the PEP-Based Design Criterion}
To provide an information-theoretic interpretation of the PEP-based criterion in Section III-C, we further examine the AMI of the uplink Rician SCMA system following analytical approaches similar to those in \cite{AMI_C} and \cite{AMI_Zilong}. Let \(\mathcal S=\{\mathbf S_1,\ldots,\mathbf S_{M^L}\}\) denote the set of possible transmitted signal matrices, where each \(\mathbf S\in\mathcal S\) corresponds to a joint selection of one codeword from each of the \(L\) UEs. Under equiprobable signaling,
\begin{equation}
	\label{input_entropy}
	p(\mathbf{S})=\frac{1}{M^L},
	\qquad
	\mathcal{H}(\mathbf{S})=L\log_2M.
\end{equation}

For a given channel realization $\mathbf{H}\equiv[\mathbf{h}_1,\mathbf{h}_2,\cdots,\mathbf{h}_K]^{\top}$, let $I_{\mathbf H}(\mathbf S;\mathbf r)$ denote the instantaneous MI between $\mathbf{S}$ and the received signal vector $\mathbf{r}\equiv[r_1,r_2,\cdots,r_K]^{\top}$, which is given by \cite{AMI_C}\cite{AMI_Zilong}
\begin{equation}
	\label{MI_definition}
	\begin{aligned} &
		I_{\mathbf H}(\mathbf S;\mathbf r)\\
		&=
		L\log_2M
		-\frac{1}{M^L}
		\sum_{\mathbf{S}\in\mathcal{S}}
		\mathrm{E}_{\mathbf{r}\mid\mathbf{S},\mathbf{H}}
		\left[
		\log_2
		\left(
		\sum_{\hat{\mathbf{S}}\in\mathcal{S}}
		\frac{
			p(\mathbf{r}|\hat{\mathbf{S}},\mathbf{H})
		}{
			p(\mathbf{r}|\mathbf{S},\mathbf{H})
		}
		\right)
		\right].
	\end{aligned}
\end{equation}
The corresponding AMI is obtained by averaging the instantaneous MI over the channel realizations as
\begin{equation}
	\label{AMI_exact}
	I_{\mathrm{AMI}}
	\triangleq
	\mathrm{E}_{\mathbf{H}}
	\left[
	I_{\mathbf H}(\mathbf S;\mathbf r)
	\right].
\end{equation}
For notational convenience, define $\lambda\triangleq\frac{E_{\mathrm{s}}}{d_{\mathrm f}N_0}$. An upper bound on the AMI is given by the following lemma.
\begin{lemma}\label{AMI_UB}
	For the uplink Rician SCMA system considered herein, the AMI is upper-bounded by $I_{\mathrm{AMI}}\leq I_{\mathrm{UP}}$, where
	\begin{equation}\label{AMI_UP_lamma1}
		\begin{aligned}
		   &I_{\mathrm{UP}}
			=L\log_2M\\
			&-\frac{1}{M^L} \sum_{\mathbf S\in\mathcal S} \log_2 \Bigg[ \sum_{\hat{\mathbf S}\in\mathcal S} \exp \Bigg( -\lambda \sum_{k\in\rho(\mathbf S,\hat{\mathbf S})} \left( |\mu_k|^2+2\sigma_k^2 \right) \Bigg) \Bigg]. 
		\end{aligned} 
	\end{equation}
\end{lemma}
\begin{IEEEproof}
	See Appendix A-1.
\end{IEEEproof}
The corresponding lower bound is given by the following lemma.
\begin{lemma}\label{AMI_LB}
	For the uplink Rician SCMA system considered herein, the AMI is lower-bounded by $I_{\mathrm{AMI}}\geq I_{\mathrm{LB}}$, where
	\begin{equation}\label{AMI_LB_lamma2}
		\begin{aligned}
			I_{\mathrm{LB}}
			&=2L\log_2M
			-K(\log_2e-1)\\
			&-\log_2\!\left[
			\sum_{\mathbf S\in\mathcal S}
			\sum_{\hat{\mathbf S}\in\mathcal S}
			\prod_{k\in\rho(\mathbf S,\hat{\mathbf S})}\frac{\exp\!\left[
				-\frac{\lambda|\mu_k|^2}
				{2(1+\lambda\sigma_k^2)}
				\right]}{1+\lambda\sigma_k^2}
			\right]. 
		\end{aligned} 
	\end{equation}
\end{lemma}
\begin{IEEEproof}
	See Appendix A-2.
\end{IEEEproof}
To relate the AMI lower bound to the PEP-based criterion derived in Section III-C, consider a distinct signal-matrix pair $\hat{\mathbf S}\neq\mathbf S$. As shown in Appendix A-3, under the high-SNR condition $\lambda\sigma_k^2\gg1$ for all $k\in\rho(\mathbf S,\hat{\mathbf S})$, the corresponding pairwise term can be approximated as
\begin{equation}\label{AMI_high_SNR_lamma}
	\prod_{k\in\rho(\mathbf S,\hat{\mathbf S})}\frac{\exp\!\left[
		-\frac{\lambda|\mu_k|^2}
		{2(1+\lambda\sigma_k^2)}
		\right]}{1+\lambda\sigma_k^2}\simeq	\frac{2 \lambda^{-\delta}}
		{d_c^2(\mathbf S,\hat{\mathbf S})}.
\end{equation}
This result reveals that the same diversity order $\delta$ and critical distance $d_c^2(\mathbf S,\hat{\mathbf S})$ governing the high-SNR PEP in \eqref{up_RicianUB_1} also characterize the codebook-dependent pairwise contributions to the AMI lower bound. Specifically, the minimum diversity order determines the dominant asymptotic behavior, while, for a given diversity order, signal-matrix pairs with smaller critical distances contribute more significantly. Hence, the PEP-based design criterion is asymptotically consistent with that implied by the AMI lower bound.

Accordingly, if all distinct signal-matrix pairs could be incorporated into the PEP-based design objective, suppressing their dominant pairwise-error contributions would also improve the codebook-dependent terms governing the AMI lower bound in the high-SNR regime, as established in Appendix A-3. This demonstrates the consistency between the union-bound and AMI-lower-bound perspectives when the complete set of pairwise error events is taken into account. In practice, however, exhaustive enumeration of all signal-matrix pairs is computationally prohibitive. Therefore, in the following subsection, we develop a practical loss-function formulation that explicitly incorporates the most significant signal-matrix pairs while accounting for the effects of higher-order error events through an additional distance-based term.

\subsection{Proposed Loss Function for Rician Fading Channels}

In the previous work \cite{Proddis}, which focuses on Rayleigh fading channels, the union-bound analysis indicates that signal-matrix pairs with $\delta > d_\mathrm{v}+1$ generally lead to negligible PEP values in the high-SNR regime. As a result, their impact on the overall union bound of the error probability becomes insignificant. However, this property does not necessarily hold under Rician fading conditions. For clarity of exposition, we begin by considering the case in which the DNN loss function is constructed solely from the PEP terms associated with pairs satisfying $\delta \leq d_\mathrm{v}+1$.

\subsubsection{Loss Function Formulation Incorporating Terms with $\delta \leq d_\mathrm{v}+1$}

The loss function can be formulated as
\begin{equation}\label{Rician_loss2}
	\begin{split}
		\Delta_{\mathrm{c}}= \Gamma_{\mathrm{c}}^{(d_\mathrm{v})}+\alpha\cdot\Gamma_{\mathrm{c}}^{(d_\mathrm{v}+1)} ,
	\end{split}
\end{equation}
where $\Gamma_{\mathrm{c}}^{(\delta)}=\sum_n \gamma^{(\delta)}_{\mathrm{c},n}$, and $0 \leq \alpha \leq 1$ denotes a hyperparameter that adjusts the relative contribution of the cases corresponding to $\delta = d_\mathrm{v}$ and $\delta = d_\mathrm{v}+1$. From an error-performance perspective, $\Gamma_{\mathrm{c}}^{(d_\mathrm{v})}$ accounts for the minimum-diversity error events and therefore captures the dominant high-SNR error behavior, whereas $\alpha\Gamma_{\mathrm{c}}^{(d_\mathrm{v}+1)}$ further accounts for the next-order multiuser error events. Minimizing these terms reduces the corresponding dominant PEP contributions and hence improves the resulting BER performance.

However, as shown in (\ref{CD_2}) and (\ref{CD_3}), the values of $\gamma^{(\delta)}_{\mathrm{c},n}$ are governed by exponential decay factors whose rates are proportional to $\kappa$. Consequently, when evaluating the loss function based on (\ref{Rician_loss2}), terms with extremely small magnitudes may lead to numerical instability under floating-point arithmetic, particularly for large values of $\kappa$. 
Based on these considerations, a log-domain formulation is proposed in the following to compute the loss function more reliably.

We first rewrite $\gamma^{(\delta)}_{\mathrm{c},n}$ as $\gamma^{(\delta)}_{\mathrm{c},n}=A_{n}^{(\delta)}\exp (-\kappa B_{n}^{(\delta)})$, where $A_{n}^{(\delta)}\equiv \left\{2\sigma^{2\delta}\prod_{k\in\rho(\mathbf{S},\mathbf{\hat{S}})}\left[\sum_{t\in \varrho(\mathbf{s}_{k},\hat{\mathbf{s}}_{k})}\left|s_{t,k}-\hat{s}_{t,k}\right|^{2}\right]\right\}^{-1}$ and $B_{n}^{(\delta)}\equiv \sum_{k\in\rho(\mathbf{S},\mathbf{\hat{S}})} \frac{\mid \sum_{t\in \varrho(\mathbf{s}_{k},\hat{\mathbf{s}}_{k})}(s_{t,k}-\hat{s}_{t,k})\mid^2}{\sum_{t\in \varrho(\mathbf{s}_{k},\hat{\mathbf{s}}_{k})}\left|s_{t,k}-\hat{s}_{t,k}\right|^{2}}$. Accordingly, we have
\begin{equation}\label{log_loss_1}
	\begin{split}
		\begin{aligned}
		\gamma^{(\delta)}_{\mathrm{c},n}&=A_{n}^{(\delta)}\exp(-\kappa B_{n}^{(\delta)})\\
		&=\exp\left[\kappa\left(\frac{1}{\kappa}\ln A_{n}^{(\delta)}\right)\right]\exp(-\kappa B_{n}^{(\delta)})\\
		&=\exp\left[-\kappa\left(B_{n}^{(\delta)}-\frac{1}{\kappa}\ln A_{n}^{(\delta)}\right)\right]\\
		&\equiv \exp(\tilde{\gamma}^{(\delta)}_{\mathrm{c},n}),
		\end{aligned}
	\end{split}
\end{equation}
where $\tilde{\gamma}^{(\delta)}_{\mathrm{c},n}$ denotes the log-domain metric of $\gamma^{(\delta)}_{\mathrm{c},n}$. Consequently, rather than calculating $\Gamma_{\mathrm{c}}^{(\delta)}$, we propose to calculate its corresponding log-domain metric
\begin{equation}\label{log_loss_2}
	\begin{split}
		\begin{aligned}
			 \tilde{\Gamma}_{\mathrm{c}}^{(\delta)}&\equiv\ln \Gamma_{\mathrm{c}}^{(\delta)}=\ln\sum_n \gamma^{(\delta)}_{\mathrm{c},n}\\
			 &=\ln\sum_n \exp(\tilde{\gamma}^{(\delta)}_{\mathrm{c},n}).
		\end{aligned}
	\end{split}
\end{equation}
Since \(\ln(\cdot)\) is monotonic, the log-domain transformation preserves the underlying PEP-based optimization objective while enabling a numerically robust evaluation when the exponential terms become extremely small. To facilitate the subsequent derivation, we first recall the max-star operator, also known as the Jacobian logarithm \cite{Robertson1995}. For any $a,b\in \mathbb{R}$, by defining $\mathrm{max}^*(a,b) \triangleq \ln(e^a +e^b)$, we have
\begin{equation}\label{log_loss_2}
	\mathrm{max}^*(a,b) = \mathrm{max}(a,b) + \ln({1+ e^{-\mid a-b \mid})}.
\end{equation}

Consequently, $\tilde{\Gamma}_{\mathrm{c}}^{(\delta)}$ can be computed recursively by applying the $\mathrm{max}^*$ operation to the previously obtained result and the next $\tilde{\gamma}^{(\delta)}_{\mathrm{c},n}$ value. Following the same line of reasoning, the detailed procedure for generating the log-domain loss function, which is defined as $\tilde{\Delta}_{\mathrm{c}}$, is summarized in Algorithm~1.


Consequently, the designed MCBs for uplink Rician fading channels, which can be obtained at the output of the DNN, are obtained by minimizing the loss function across multiple iterations or updates, as denoted by
\begin{equation}\label{ULRician_criterion}
	\begin{split}
		\{\mathcal{C}^{*}_{1},\mathcal{C}^{*}_{2},\cdots,\mathcal{C}^{*}_{L}\}=\arg\min_{\mathbf{c}^{(t)}, t = 1,2, \cdots , T}\tilde{\Delta}_{\mathrm{c}}(\mathbf{c}^{(t)}),
	\end{split}
\end{equation}
where $t$ denotes the iteration index for updating the generated signal points $\mathbf{c}^{(t)}$, and $T$ represents the maximum number of iterations. 
This formulation steers the gradient-based optimization toward smaller loss values and, equivalently, toward codebooks with more favorable critical-distance properties. As a result, the codebooks are progressively refined, leading to improved alignment with the underlying transmission environment and, consequently, enhanced system robustness and error performance.


\subsubsection{Extended Loss Function Formulation Incorporating Terms with $\delta > d_\mathrm{v}+1$}

As the Rician $K$-factor $\kappa$ increases, which in practice corresponds to a decrease in $\sigma^2$, the factor $2\sigma^{2\delta}$ in (\ref{CD_3}) decays rapidly. On the other hand, under the high-SNR assumption, the term $\left(\frac{E_{\mathrm{s}}}{2d_{\mathrm{f}}N_\mathrm{0}}\right)^{-\delta}$ in (\ref{up_RicianUB_1}) is also significantly attenuated. Hence, the effects of these two terms tend to offset each other. As a result, the relative influence of SD becomes less pronounced in the large-$\kappa$ regime, whereas the WSED-related exponential term in (\ref{CD_3}) becomes increasingly dominant. Consequently, error events with \(\delta>d_{\mathrm v}+1\), which are generally negligible under Rayleigh fading in the high-SNR regime, may make non-negligible contributions to the union bound and should therefore be accounted for in the codebook design.


\begin{algorithm}[t]
	\caption{Calculation of the Loss Function in Log-Domain}\label{algorithm_loss}
	\begin{algorithmic}[1]
		\STATE Initialize $\tilde{\Delta}_{\mathrm{c}}=0$
		\FOR{$\delta = d_\mathrm{v} :d_\mathrm{v}+1$}
		\IF{$\delta = d_\mathrm{v}$}  
		\STATE Calculate a total of $LU$ $\tilde{\gamma}^{(d_\mathrm{v})}_{\mathrm{c},n}$ values indexed by $n$ based on (\ref{CD_2}) and (\ref{log_loss_1}).
		\STATE Let $\tilde{\Delta}_{\mathrm{c}}=\tilde{\gamma}^{(d_\mathrm{v})}_{\mathrm{c},1}$.
		\FOR  {$n = 2 :LU$}
		\STATE Update $\tilde{\Delta}_{\mathrm{c}}$ by $\tilde{\Delta}_{\mathrm{c}}= \mathrm{max}^*(\tilde{\Delta}_{\mathrm{c}},\tilde{\gamma}^{(d_\mathrm{v})}_{\mathrm{c},n})$.
		\ENDFOR
		
		\ELSIF{$\delta = d_\mathrm{v}+1$}
		\FOR {$e=1 : \tbinom{K}{d_\mathrm{v}+1}$}
		\STATE Based on $\mathbf{F}$, find the $L_e$ UEs that only utilize the $e$-th combination of $d_\mathrm{v}+1$ REs.
		\STATE Generate a total of $\tbinom{M^{L_e}}{2}$ $\tilde{\gamma}^{(d_\mathrm{v}+1)}_{\mathrm{c},n}$ values indexed by $n$ based on (\ref{CD_3}) and (\ref{log_loss_1}).
		\STATE Similar to (\ref{log_loss_1}), we have 
		\begin{equation}\label{Updated_Loss_Proposition1}
			\begin{split}
				\alpha\gamma^{(d_\mathrm{v}+1)}_{\mathrm{c},n}&=\exp(\ln\alpha+\tilde{\gamma}^{(d_\mathrm{v}+1)}_{\mathrm{c},n})\\
			\end{split}
		\end{equation}
		\FOR  {$n = 1 :\tbinom{M^{L_e}}{2}$}
		\STATE Update $\tilde{\Delta}_{\mathrm{c}}$ by $\tilde{\Delta}_{\mathrm{c}}= \mathrm{max}^*(\tilde{\Delta}_{\mathrm{c}},\ln\alpha+\tilde{\gamma}^{(d_\mathrm{v}+1)}_{\mathrm{c},n})$
		\ENDFOR
		
			\ENDFOR
			\ENDIF
			%
			%
			\ENDFOR
			\STATE \bf{Output:} $\tilde{\Delta}_{\mathrm{c}}$
		\end{algorithmic}
	\end{algorithm}


However, directly incorporating all such higher-order error events into the loss function is computationally prohibitive because of the extremely large number of possible signal-matrix pairs. We therefore introduce an additional loss term based on the {\it minimum single-RE weighted squared Euclidean distance} (MSWSED), which characterizes the worst-case WSED among the superimposed signal-point pairs on each RE without explicitly enumerating all higher-order signal-matrix pairs. Let $p$ and $q$ denote the indices of two superimposed signal points in the $k$-th RE, and let ${x}_{l_{k,i},k}^{(p)}$ represent the signal transmitted by UE $l_{k,i}$ corresponding to the $p$-th superimposed point on RE $k$. The MSWSED of the $k$-th RE is defined as
\begin{equation}\label{MSWSED}
	\begin{split}
		d^2_{\mathrm{MSWSED},k} = \min_{\forall p,q} \frac{\left| \sum_{i=1}^{d_{\mathrm{f}}}  {x}_{l_{k,i},k}^{(p)} - \sum_{i=1}^{d_{\mathrm{f}}} {x}_{l_{k,i},k}^{(q)} \right|^2}{\sum_{i=1}^{d_{\mathrm{f}}}\left|{x}_{l_{k,i},k}^{(p)}-{x}_{l_{k,i},k}^{(q)}\right|^2}.
	\end{split}
\end{equation}

Accordingly, we re-define the extended loss function as
\begin{equation}\label{Rician_loss3}
	\begin{split}
		\Delta_{\mathrm{c,ext}}= \Gamma_{\mathrm{c}}^{(d_\mathrm{v})}+\alpha\cdot\Gamma_{\mathrm{c}}^{(d_\mathrm{v}+1)}+\beta\cdot e^{ -\kappa\sum_{k=1}^{K}d^2_{\mathrm{MSWSED},k}}
	\end{split}
\end{equation}
and the corresponding log-domain extended loss function as
\begin{equation}\label{Rician_loss3}
	\begin{split}
		\tilde{\Delta}_{\mathrm{c,ext}}=\mathrm{max}^*\left(\tilde{\Delta}_{\mathrm{c}},\ln\beta-\kappa\sum_{k=1}^{K}d^2_{\mathrm{MSWSED},k}\right),
	\end{split}
\end{equation}
where \(0 \leq \beta \leq 1\) denotes a hyperparameter that controls the contribution of the additional term to the extended loss function. The summation of the MSWSED values accounts for higher-order error events with \(\delta>d_{\mathrm v}+1\) without requiring explicit enumeration of all corresponding signal-matrix pairs. Since the MSWSED on each involved RE provides a lower bound on the corresponding WSED contribution, increasing the MSWSED values raises these lower bounds across different error events. Consequently, the associated higher-order PEP contributions, which become increasingly important in the strong-LOS regime, can be suppressed, thereby improving the BER performance. Accordingly, the MSWSED values are computed and incorporated into the loss evaluation in Algorithm 1 for subsequent DNN-based codebook optimization.

\subsection{Codebook Design for AWGN and Uplink Rayleigh Fading Channels}

The uplink Rayleigh fading channel can be regarded as a special case of the uplink Rician fading channel with $\kappa = 0$, where the LOS component is absent. Without loss of generality, by setting $\mu = 0$ and $\sigma^2 = 1$, the CD expression in (\ref{CD_3}) reduces to
\begin{equation}\label{CD_4}
	\begin{aligned}
			d^2_{\mathrm{c}}(\mathbf{S},\mathbf{\hat{S}})&=2\prod_{k\in\rho(\mathbf{S},\mathbf{\hat{S}})}\left[\sum_{t\in \varrho(\mathbf{s}_{k},\hat{\mathbf{s}}_{k})}\left|s_{t,k}-\hat{s}_{t,k}\right|^{2}\right],
		\end{aligned}
\end{equation}
which is identical to the well-known {\it product distance} (PD) expression for fast Rayleigh fading channels reported in the literature \cite{Proddis}, \cite{STTC_Vucetic}.

As discussed in \cite{Proddis}, SCMA codebook design for uplink Rayleigh fading channels can be interpreted as the construction of a $d_\mathrm{f}\times K$ space-time code for fast Rayleigh fading channels, where the design criterion focuses on maximizing the minimum PD (MPD) of the overall signal matrix when $\delta=d_\mathrm{v}$. 
It is shown through (\ref{up_RicianUB_1}) that the contributions of terms with $\delta>d_{\mathrm{v}}$ are strongly suppressed by the SNR-dependent factor $\left( {\frac{E_{\mathrm{s}}}{2d_{\mathrm{f}} N_\mathrm{0}}}\right)^{-\delta}$ in the union bound under the high-SNR assumption. Since cases with $\delta > d_{\mathrm{v}}$ occur only when more than one UE transmits distinct codewords in the corresponding signal-matrix pair, only a single MCB needs to be designed and can then be shared by all UEs in the uplink Rayleigh fading case. This design choice avoids UE-dependent codebook optimization and simplifies the DNN architecture.
For single-MCB design, the DNN output dimension is reduced to $I_\mathrm{c}=2Md_\mathrm{v}$. In addition, Algorithm 1 evaluates the loss function only for the cases with $\delta=d_\mathrm{v}$, i.e., $\tilde{\Delta}_{\mathrm{c}}= \tilde{\Gamma}_{\mathrm{c}}^{(d_\mathrm{v})}$. 

On the other hand, as the Rician $K$-factor $\kappa$ increases, the LOS component becomes dominant and the effect of multipath fading gradually diminishes. In the large-$\kappa$ regime, the channel behavior approaches that of a deterministic channel, and the system performance is mainly limited by additive noise. Therefore, the AWGN channel can be regarded as a limiting case of the Rician fading channel as $\kappa \rightarrow \infty$. In this case, the channel randomness vanishes, allowing us to set $\sigma^2 = 0$ and normalize the deterministic channel gain to $\mu = 1$. By substituting these parameters into (\ref{up_RicianMV}) and (\ref{up_RicianUPEP_1}), the PEP for the AWGN channel can be derived as
\begin{equation}
	\begin{split}
		\label{up_AWGNPEP}
		\begin{aligned}
			P(\mathbf{S} &\rightarrow \mathbf{\hat{S}}) \leq \frac{1}{2}  \exp\left(-\frac{E_{\mathrm{s}}}{4d_{\mathrm{f}}N_\mathrm{0}} \sum\limits_{k\in\rho(\mathbf{S},\mathbf{\hat{S}})}\left| \sum\limits^{d_\mathbf{f}}_{i=1} \Delta x_{l_{k,i},k}  \right|^2 \right)\\
			&\equiv \frac{1}{2}  \exp\left(-\frac{E_{\mathrm{s}}}{4d_{\mathrm{f}}N_\mathrm{0}} \norm{ \mathbf{y}^{(p)}-\mathbf{y}^{(q)}  }^2 \right),
		\end{aligned}
	\end{split}
\end{equation}
which is identical to the PEP between two length-$K$ superimposed signal vectors, denoted by $\mathbf{y}^{(p)}$ and $\mathbf{y}^{(q)}$. The $p$-th superimposed signal vector $\mathbf{y}^{(p)}$ is defined as $\mathbf{y}^{(p)}\equiv[\sum_{i=1}^{d_\mathrm{f}} x_{l_{1,i},1}^{(p)}, \sum_{i=1}^{d_\mathrm{f}}x_{l_{2,i},2}^{(p)}, \ldots, \sum_{i=1}^{d_\mathrm{f}}x_{l_{K,i},K}^{(p)}]$, with $M^L$ possible vectors in total. The SED between superimposed signal vectors, denoted as $\norm{ \mathbf{y}^{(p)}-\mathbf{y}^{(q)}  }^2$, governs the PEP. Consequently, to improve the BER performance in the high-SNR region, the MSED among all superimposed signal vector pairs should be maximized.




Since gradient descent minimizes the loss function, the loss should be inversely related to the MSED among superimposed signal-vector pairs. The enhanced metric in \cite{Proddis} further accounts for the nearest superimposed signal-vector pairs with small SED values. Motivated by these considerations, let \(\mathfrak{y}\) denote the index set of superimposed signal-vector pairs corresponding to the $\Upsilon$ smallest SED values. Accordingly, the AWGN loss function is defined as
\begin{equation}\label{AWGN_loss_upd}
	\begin{split}
		\begin{aligned}
		\Gamma_{\mathrm{AWGN}}&= \sum_{(p,q)\in \mathfrak{y}}\exp{\left(-\omega\norm{ \mathbf{y}^{(p)}-\mathbf{y}^{(q)}  }^2\right)}\\
		&\equiv \sum_{(p,q)\in \mathfrak{y}}\exp{\left(\gamma_{\mathrm{AWGN},p,q}\right)},
		\end{aligned}
	\end{split}
\end{equation}
where $\omega$ is chosen to be sufficiently large, indicating that the design targets the high-SNR regime. Although $\kappa$ is eliminated in (\ref{up_AWGNPEP}) and (\ref{AWGN_loss_upd}), the large value of $\omega$ may still lead to extremely small terms in the AWGN design metric. As a result, log-domain computation remains necessary and can be implemented by sequentially applying the max-star operator to the $\gamma_{\mathrm{AWGN},p,q}$ terms. The resulting log-domain loss function for the AWGN channel, denoted by $\tilde{\Gamma}_{\mathrm{AWGN}}\equiv \ln\Gamma_{\mathrm{AWGN}}$, is then used for DNN training.

Since superimposed signal vector pairs yielding small SED values are not necessarily restricted to cases with small SD values, a joint optimization of the MCBs across all UEs becomes necessary. Accordingly, the number of output neurons of the DNN is again set to $I_\mathrm{c}=2LMd_\mathrm{v}$. 
It is worth emphasizing that, compared with the design process in \cite{Proddis}, the proposed framework eliminates explicit signal sub-constellation selection and separate rotation-angle optimization, since these design aspects are inherently embedded in the generation of $\mathbf{c}^{(t)}$ through the DNN. Moreover, the procedures referred to as single-user observation (SUO) and multiple-user observation (MUO) reported in \cite{Proddis} are jointly optimized within a unified fully connected neural network.


Finally, after completing the codebook design process, the bit-labeling rule for each MCB is determined using the binary switching algorithm (BSA) adopted in our earlier work \cite{Proddis}. The BSA iteratively improves the bit-labeling rule by swapping the binary labels assigned to the mother codewords and retaining only those swaps that reduce the labeling cost. Through this procedure, codeword pairs with higher pairwise error probabilities are preferentially assigned labels with smaller Hamming distances, consistent with the general principle of Gray labeling. 

Due to the large number of possible transmitted signal matrices, jointly optimizing the labeling rules of all UE codebooks would incur substantial computational complexity. Therefore, the labeling rule of each UE codebook is optimized separately, effectively considering only the signal-matrix pairs with \(\delta=d_{\mathrm{v}}\). This provides a practical tradeoff between optimization performance and complexity, while joint labeling and codebook optimization is left for future investigation.

	\section{Numerical Results}
	We consider a system configuration of 6 UEs, 4 REs, an overloading rate of 150$\%$, $d_\mathrm{f}=3$, and $d_\mathrm{v}=2$. The maximum Doppler shift and the angle of arrival are respectively set as $f_\mathrm{d} = 100$ Hz and $\theta_0 = 30^o$, and $\kappa$ varies from 0 to 20. 
	The $\frac{E_{\mathrm{b}}}{N_\mathrm{0}}$ value of the system is defined as
	\begin{equation}\label{SNR}
		\begin{split}
			\frac{E_{\mathrm{b}}}{N_\mathrm{0}}=\frac{E_{\mathrm{s}}}{N_\mathrm{0}}\cdot \frac{K}{L\cdot\log_2M},
		\end{split}
	\end{equation}
	where $\frac{E_{\mathrm{s}}}{N_\mathrm{0}}$ denotes the SNR value of each RE.
	The number of iterations within the MPA receiver is set to 5.

	
	We first present the performance results over AWGN channels. The SCMA MCBs designed under the power-balanced constraint are provided in Appendix~B-1, and the corresponding signal sub-constellations of each UE on each RE are shown in Fig.~\ref{AWGN_M4_cons}. It can be observed that the signal points are irregularly distributed across different REs. Focusing on UE~1, very closely spaced signal points are assigned on RE~1, causing some superimposed signal points on this RE to become nearly overlapped.
	However, a closer inspection of Appendix~B-1 reveals that the two mother codewords of UE~1 associated with the smallest SED on RE~1 (i.e., the first and second columns) correspond to signal points with a large SED on RE~2. Hence, a sufficiently large SED between the mother codewords can still be preserved. This observation is consistent with the design philosophy in \cite{LP}, where superimposed signal points with zero Euclidean distance are allowed on individual REs, thereby providing greater flexibility in the overall codebook design.


	\begin{figure}[!t]
		\centering
		\includegraphics[width = 3.5in]{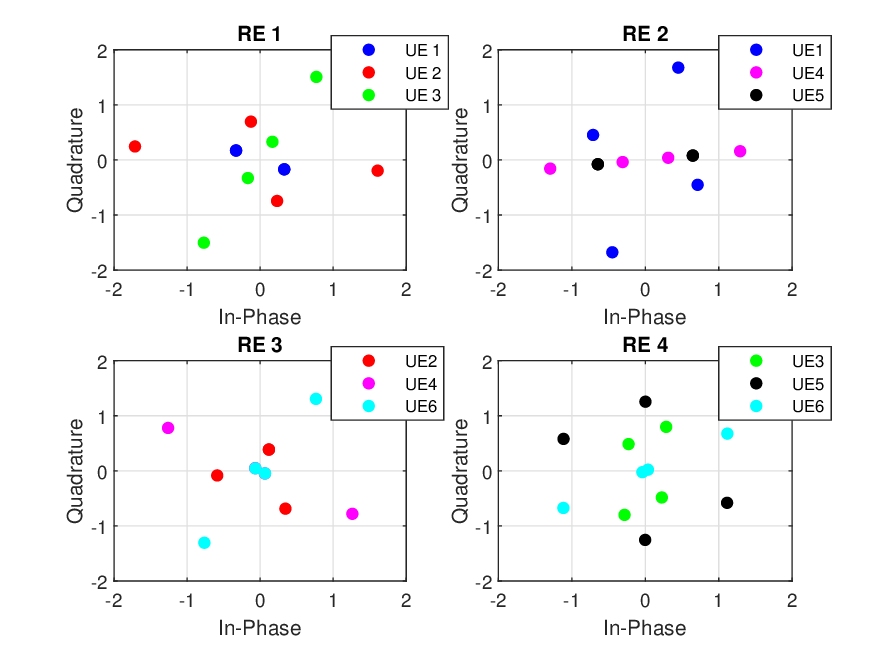}
		\caption{DNN-derived signal sub-constellations on AWGN channels where $M=4$.}
		\label{AWGN_M4_cons}
		\vspace{-0.4cm}
	\end{figure}

	\renewcommand\arraystretch{1.3}
	\tabcolsep=9.2pt
	\begin{table*}\label{Table_AWGN}
			\centering
			\caption{Comparison of design metric values for various schemes over AWGN channels where $M=4$.}
			\begin{tabular*}{18.2cm}{ p{0.8cm}  p{1.01cm} p{1.01cm}  p{0.87cm} p{1.01cm} p{1.05cm} p{1.01cm} p{1.14cm} p{1.01cm} p{1.01cm} p{1.3cm}}
				\toprule
				$ $ & $\mathrm{Prop.}$ & $\mathrm{Prop.}^{\dagger}$ & CB \cite{cpa} & LP$^{\dagger}$ \cite{LP} & DE$^{\dagger}$ \cite{DE} & GA \cite{GA} & HUA \cite{hua} & PI$^{\dagger}$ \cite{PI} & SU \cite{SU} & CHEN \cite{Proddis} \\
				\midrule
				MSED & 3.3816 & 3.5487 & 1.3428 & 3.0016 & 1.7764 & 2.25 & 0.6293 & 2.5931 & 3.3654 & 2.2928 \\
				$\tilde{\Gamma}_{\mathrm{AWGN}}$ & -162.726 & -170.915 & -62.873 & -143.172 & -86.741 & -105.605 & -27.313 & -122.756 & -161.375 & -109.618 \\
				\bottomrule
			\end{tabular*}
			\begin{tablenotes}
				\footnotesize
				\item The superscript $\dagger$ indicates power-imbalanced designs.
			\end{tablenotes}
	\end{table*}

	Several SCMA schemes designed for AWGN channels are selected as benchmark schemes for comparison. As shown in Table~II, the proposed schemes under both power-balanced and power-imbalanced settings achieve the smallest $\tilde{\Gamma}_{\mathrm{AWGN}}$ among all benchmark schemes. Moreover, it is noteworthy that the proposed schemes also attain the largest MSED values among all superimposed signal vector pairs, even though MSED is not explicitly adopted as the design metric.

Consistent with these favorable metric values, Fig.~\ref{AWGN_M4} shows that the proposed codebooks achieve competitive performance compared with existing schemes, including that reported in our earlier work \cite{Proddis}. The biconvex-optimized codebooks in \cite{SU} slightly outperform the proposed design under the considered AWGN setting. The proposed DNN-assisted framework therefore represents an alternative approach to SCMA codebook design, with the advantage of jointly optimizing multiple codebooks within a unified architecture adaptable to different system configurations and channel conditions.
	

	
	\begin{figure}[!t]
		\centering
		\includegraphics[width = 3.5in]{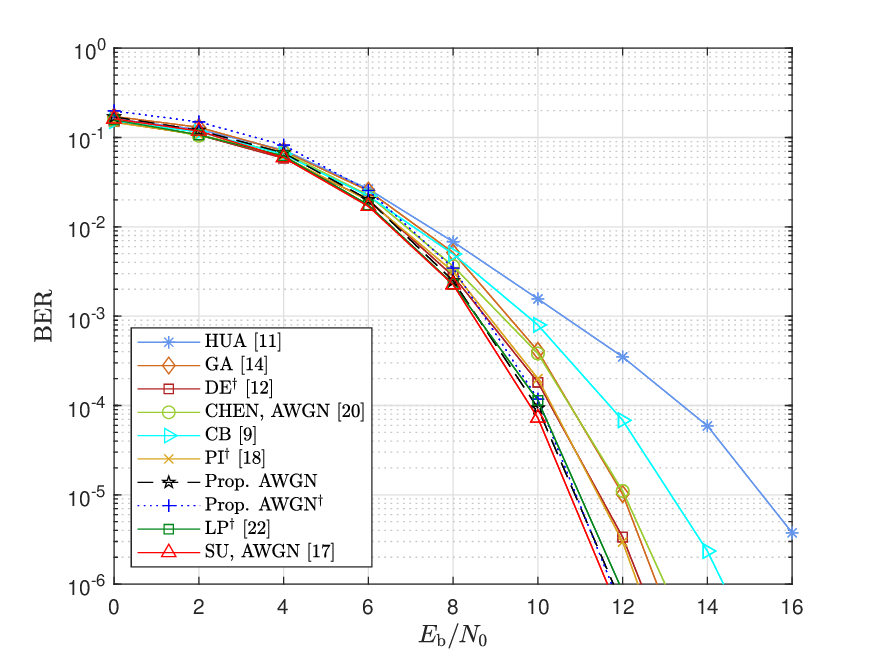}
		\caption{BER performance on the AWGN channel where $M=4$. The superscript $\dagger$ indicates power-imbalanced designs.}
		\label{AWGN_M4}
		\vspace{-0.4cm}
	\end{figure}
	
	\begin{figure}[!t]
		\centering
		\includegraphics[width = 3.5in]{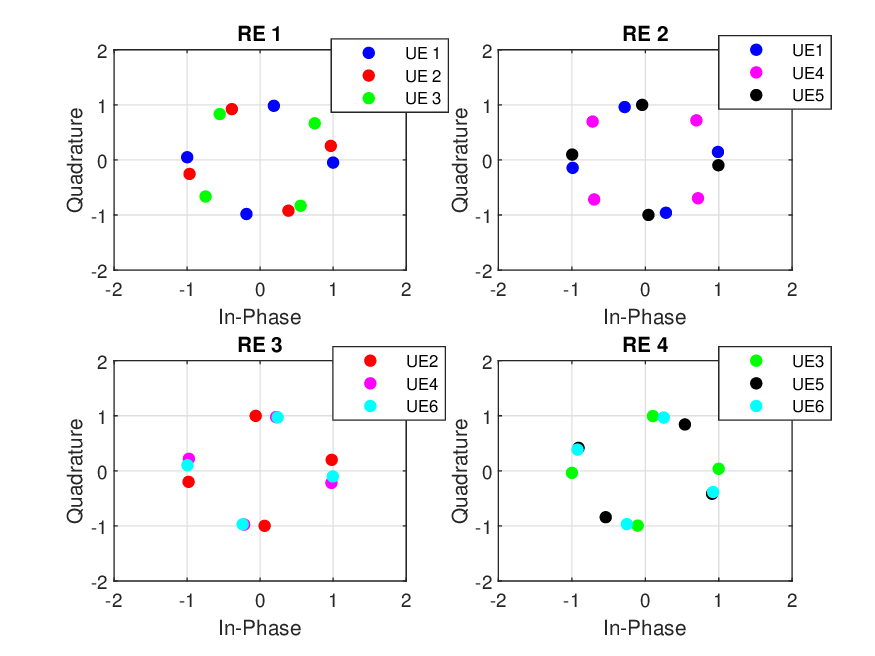}
		\caption{DNN-derived signal constellations for uplink Rayleigh fading channels where $M=4$.}
		\label{Rayleigh_M4_cons}
		\vspace{-0.4cm}
	\end{figure}

	\renewcommand\arraystretch{1.3}
	\tabcolsep=9.2pt
	\begin{table*}\label{Table_ul}
		\centering
		\caption{Comparison of $\tilde{\Delta}_{\mathrm{c}}$ values for various schemes over uplink Rayleigh fading channels.}
		\begin{tabular*}{18cm}{ p{1.1cm}  p{0.5cm} p{1.64cm} p{1.69cm} p{1.54cm} p{1.19cm} p{0.8cm} p{0.93cm} p{0.88cm} p{1.3cm}}
			\toprule
			$M=4$ & $\mathrm{Prop.}$ &  LUO \cite{VM} & \hspace{-0.2cm}CR \cite{CR} & \hspace{-0.3cm}DE \cite{DE} & \hspace{-0.4cm}GA \cite{GA} & \hspace{-0.5cm}HUA \cite{hua} & PI$^{\dagger}$ \cite{PI} & CS \cite{CS} & CHEN \cite{Proddis} \\
			
			& 3.7499 & 3.7519 & \hspace{-0.2cm}3.7748 & \hspace{-0.3cm}3.7791 & \hspace{-0.4cm}634.876 & \hspace{-0.5cm}3.8910 & 4.01 & 4.6148 & 3.7519 \\
			\midrule
			$M=8$ & $\mathrm{Prop.}$ &  Cir-QAM \cite{Proddis} & \hspace{-0.2cm}Star-QAM \cite{Proddis} & \hspace{-0.3cm}Tri-QAM \cite{Proddis} & \hspace{-0.5cm}LUO \cite{VM} &   &   & \\
			& 5.4559 & 5.5467 & \hspace{-0.2cm}5.4633 & \hspace{-0.3cm}5.4626 & \hspace{-0.5cm}5.5526 & & & \\
			\midrule
			$M=16$ & $\mathrm{Prop.}$ &  LUO \cite{VM} & \hspace{-0.2cm}PAM\cite{Proddis} & & &   &   &   & \\
			& 7.2055 & 7.4039 & \hspace{-0.2cm}51.8448 & & & & & & \\
			\bottomrule
		\end{tabular*}
			\vspace{-0.4cm}
	\end{table*}

	
Next, we consider uplink Rayleigh fading channels. For \(M=4\), our previous work \cite{Proddis} employed QPSK sub-constellations. By intentionally optimizing six distinct UE-specific codebooks, the DNN-designed sub-constellations in Fig.~\ref{Rayleigh_M4_cons} naturally exhibit similar QPSK-like and rotated structures, although they are not identical to conventional QPSK constellations. As shown in Table~III, the proposed codebooks achieve a $\tilde{\Delta}_{\mathrm{c}}$ value comparable to that of \cite{Proddis} and the smallest among the considered benchmark schemes. Since the codebooks in \cite{Proddis} already achieve a near-optimal \(\tilde{\Delta}_{\mathrm{c}}\), only a marginal BER improvement is observed in Fig.~\ref{ul_m4}.


For the single-MCB design with \(M=8\), the DNN-derived sub-constellation naturally exhibits a structure closely resembling the Triangular QAM constellation adopted in \cite{Proddis}, with correspondingly similar BER performance shown in Fig.~\ref{ul_m8_m16}a. The purpose of this comparison is primarily to demonstrate that the proposed unified DNN-assisted framework can recover near-optimal designs for the well-studied uplink Rayleigh fading case, rather than to substantially outperform existing Rayleigh-specific schemes. In this scenario, conventional constellations with well-separated signal points already exhibit favorable MPD properties, while the design procedure in \cite{Proddis} was specifically developed for Rayleigh fading and already provides near-optimal performance, leaving limited room for further BER improvement.


For clarity, the LUO curves use the codebooks reported in \cite{VM}, whereas the Cir-QAM, Star-QAM, Tri-QAM, and PAM curves use only the corresponding constellation points, with the codebooks constructed according to \cite{Proddis}. Although LUO yields a relatively small $\tilde{\Delta}_{\mathrm c}$ for $M=8$ in Table~III, its MSD is only 1, resulting in a lower diversity order and the gentler BER slope observed in Fig.~\ref{ul_m8_m16}a. Overall, Table~III and Fig.~\ref{ul_m8_m16} confirm the effectiveness of the proposed DNN-assisted design in the limiting case \(\kappa\rightarrow0\), corresponding to Rayleigh fading.

\begin{figure}[!t]
	\centering
	\includegraphics[width = 3.5in]{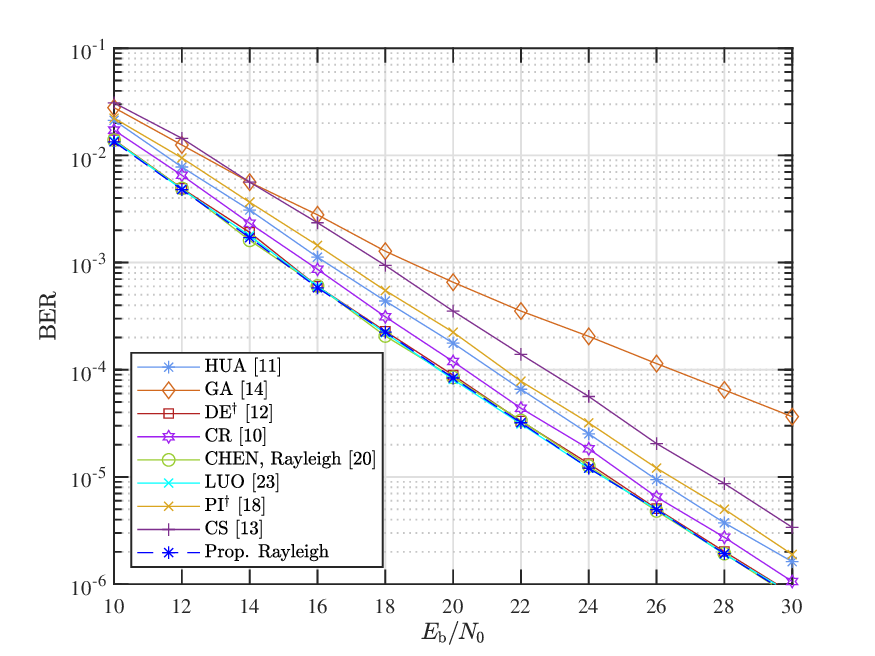}
	\caption{BER performance for uplink Rayleigh fading channels where $M=4$. The superscript $\dagger$ indicates power-imbalanced designs.}
	\label{ul_m4}
	\vspace{-0.4cm}
\end{figure}


\begin{figure}[!t]
\centering
\includegraphics[width = 3.7in]{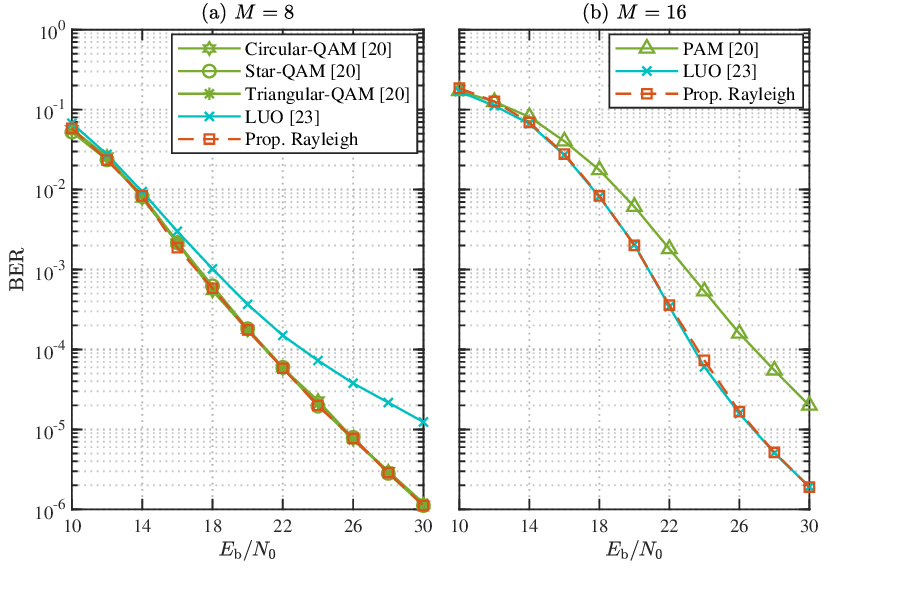}
\caption{BER performance for the uplink Rayleigh fading channel. (a) $M=8$; (b) $M=16$.}
\label{ul_m8_m16}
\vspace{-0.4cm}
\end{figure}

\begin{figure}[!t]
	\centering
	\includegraphics[width = 3.4in]{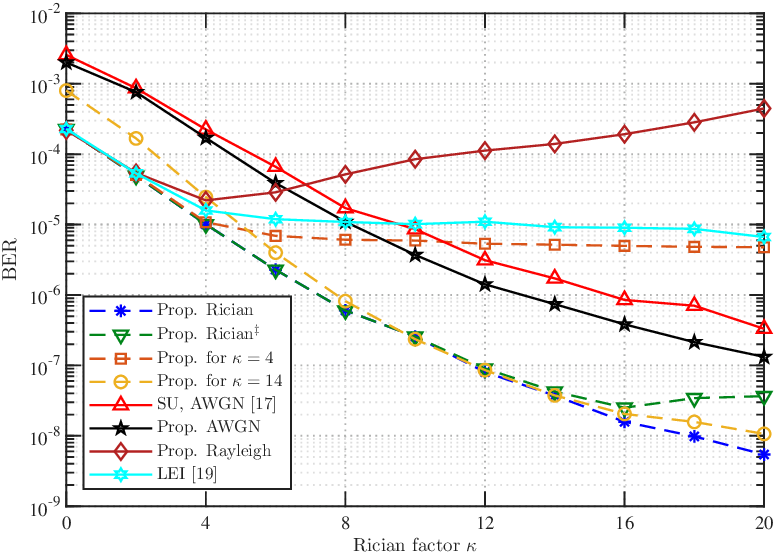}
	\caption{BER performance evaluation of SCMA codebooks over uplink Rician channels with varying $\kappa$ values at $\mathrm{\frac{E_{\mathrm{b}}}{N_\mathrm{0}}}=18$ dB. The superscript $\ddag$ indicates the codebooks designed without incorporating the MSWSED terms.}
	\label{Rician_M4}
	\vspace{-0.4cm}
\end{figure}

Focusing on a relatively high \(E_{\mathrm b}/N_0\) value of \(18\) dB, Fig.~\ref{Rician_M4} illustrates the BER performance of several SCMA codebooks as a function of the Rician factor \(\kappa\). The flattening observed in some BER curves reflects the reduced sensitivity of the corresponding codebooks to further changes in the LOS/NLOS balance as \(\kappa\) increases. The proposed Rayleigh codebooks achieve the best performance at \(\kappa=0\), but their effectiveness gradually degrades as \(\kappa\) increases. In contrast, the proposed AWGN codebooks, together with the SU codebooks reported in \cite{SU}, are expected to perform best in the limiting case of \(\kappa\rightarrow\infty\), while being less favorable in the low-\(\kappa\) regime.

	
	It is worth noting that the codebooks in \cite{Tauf} are specifically designed for uplink Rician fading channels. Instead of using the CD expressions in (\ref{CD_1})-(\ref{CD_3}), \cite{Tauf} introduces the minimum Euclidean distance across resource elements (MED-RE) and derives a simpler but looser PEP upper bound. A three-step procedure is then used to optimize the MCBs and rotation angles by mitigating the effects of small MPD and MED-RE values, their multiplicities, and small MSEDs among superimposed signal-vector pairs. However, since these metrics do not explicitly depend on \(\kappa\), their effectiveness may be limited across varying LOS conditions, especially in the high-\(\kappa\) regime.


	The green dashed curve with inverted-triangle markers in Fig.~\ref{Rician_M4} corresponds to the designed codebooks when the summation term of the MSWSEDs is omitted from the loss function in (\ref{Rician_loss3}). As discussed in Section III-E-2, higher-order error events with \(\delta>d_\mathrm{v}+1\) become increasingly relevant as \(\kappa\) grows. The MSWSED term is therefore introduced to account for their effects without explicitly enumerating all corresponding signal-matrix pairs. Consequently, its removal leads to noticeable degradation in the high-\(\kappa\) regime. The blue dashed curve with asterisk markers represents the near-optimal codebooks designed for different \(\kappa\) values, which consistently outperform the benchmark schemes over a wide range of LOS conditions. For three representative \(\kappa\) values, the corresponding \(\tilde{\Delta}_{\mathrm c}\) values are summarized in Table~IV and closely follow the BER trends in Fig.~\ref{Rician_M4}, further validating the proposed DNN-based design for Rician fading channels.

	\renewcommand\arraystretch{1.3}
	\tabcolsep=9.2pt
	\begin{table*}\label{Table_Rician}
		\centering
		\caption{Comparison of $\tilde{\Delta}_{\mathrm{c}}$ values for various schemes over Rician fading channels with $M=4$ and $\frac{E_{\mathrm{b}}}{N_\mathrm{0}}=18\mbox{ dB}$.}
		\begin{tabular*}{11cm}{ p{0.7cm} p{1.3cm}  p{1.2cm} p{1.42cm} p{1.38cm} p{1.45cm} }
			\toprule
			 $\kappa$ & $\mbox{Prop. Rician}$ & LEI \cite{Tauf} & $\mbox{Prop.  Rayleigh}$ & $\mbox{Prop. AWGN}$  & SU \cite{SU} \\
			\midrule
			4 & -5.7379 & -5.4228 & -5.2005 & 4.8150  & 2.4715 \\
		    12 & -21.0189 & -8.6674 & -5.4112 & -8.9687  & -13.9071 \\
			20 & -37.0446 & -11.5613 & -6.0180 & -12.9190  & -29.5988 \\
			%
			%
			\bottomrule
		\end{tabular*}
\vspace{-0.4cm}
	\end{table*}

To investigate the sensitivity to Rician-factor mismatch, Fig.~\ref{Rician_M4} also evaluates two fixed codebooks designed for \(\kappa=4\) and \(\kappa=14\) while varying the actual \(\kappa\). The results show that each codebook remains effective over a neighboring range of \(\kappa\) values, although the mismatch behavior is asymmetric. Performance degradation is moderate when the actual \(\kappa\) is moderately below the design value, but becomes more pronounced when it exceeds the design value, consistent with the increasing importance of the WSED-related contribution as \(\kappa\) increases. Thus, codebooks designed for larger \(\kappa\) can tolerate a moderate decrease in \(\kappa\), whereas those designed for smaller \(\kappa\) become less effective under stronger LOS conditions. These results suggest that redesign is unnecessary for small variations in \(\kappa\), but becomes increasingly important when the operating \(\kappa\) exceeds the design value.

%

	Figs.~\ref{K4_K8} and \ref{K12_K16} compare the BER performance of various codebooks for \(\kappa=4,8,12,\) and \(16\). At relatively small \(\kappa\), e.g., \(\kappa=4\), the proposed Rician-specific codebooks achieve near-optimal performance, while the Rayleigh and LEI codebooks remain competitive. By contrast, the AWGN and SU codebooks perform less favorably under weak-LOS conditions. As \(\kappa\) increases, corresponding to a stronger LOS component, the effectiveness of the Rayleigh codebooks gradually degrades, whereas the AWGN and SU codebooks become more suitable for the resulting channel conditions and exhibit progressively improved performance. The proposed method maintains near-optimal performance across all four \(\kappa\) values, with a representative MCB for \(\kappa=8\) provided in Appendix~B-2. At a BER of \(10^{-7}\), the proposed designs provide \(E_{\mathrm b}/N_0\) gains of \(4.8\), \(6.61\), and \(7.08\) dB over the LEI codebooks for \(\kappa=8,12,\) and \(16\), respectively, and corresponding gains of \(9.81\), \(6.27\), and \(4.39\) dB over the SU codebooks, demonstrating 
	the effectiveness of the proposed framework and its strong potential for ultra-reliable Rician SCMA transmission.

	

\begin{figure}[!t]
	\centering
	\includegraphics[width = 3.5in]{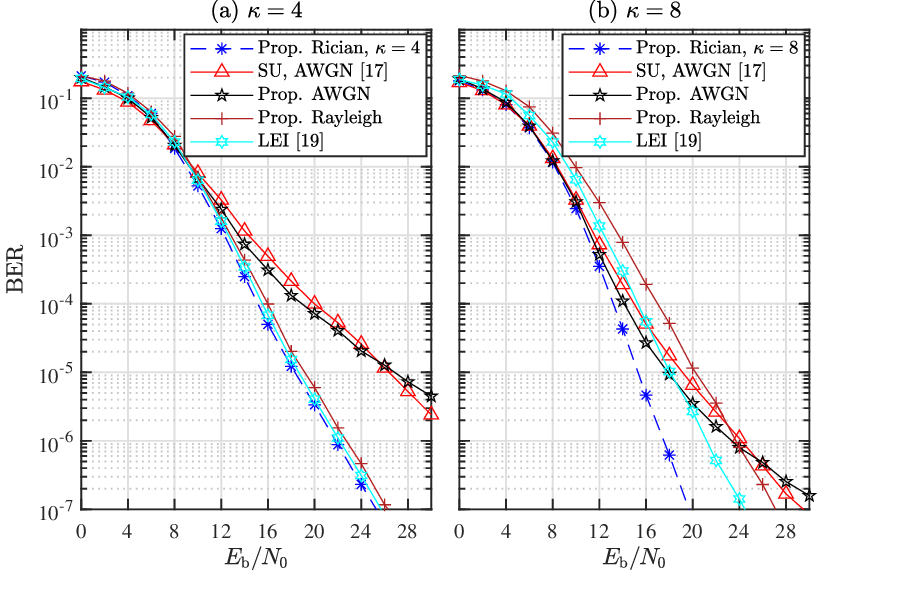}
	\caption{BER performance evaluation of SCMA codebooks over uplink Rician channels with parameters $M=4$, $\kappa=4$ and $8$.}
	\label{K4_K8}
\end{figure}

\begin{figure}[!t]
	\centering
	\includegraphics[width = 3.5in]{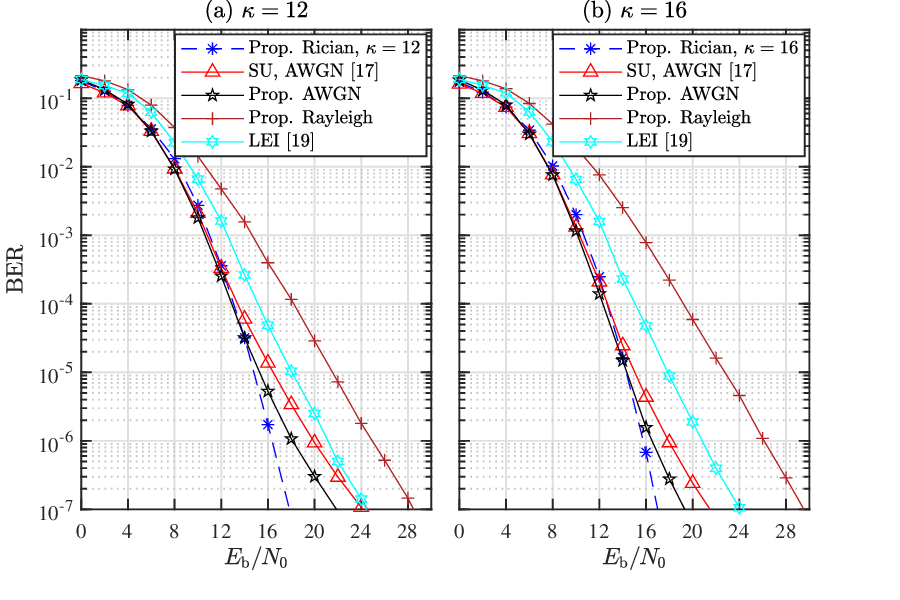}
	\caption{BER performance evaluation of SCMA codebooks over uplink Rician channels with parameters $M=4$, $\kappa=12$ and $16$.}
	\label{K12_K16}
	\vspace{-0.4cm}
\end{figure}
	
	\section{Conclusion}
	In this paper, a novel DNN-assisted framework is proposed for SCMA codebook construction. The DNN is trained in an unsupervised manner, with a loss function derived from the structural characteristics of the generated codebooks, thereby eliminating the need for labeled training data. Based on error-probability analysis, \(\kappa\)-dependent design criteria are derived for uplink Rician fading channels and used to formulate the loss function. These criteria naturally encompass AWGN and Rayleigh fading as limiting cases, providing a unified framework for codebook design. In addition, upper and lower bounds on the average mutual information (AMI) are derived, providing an information-theoretic interpretation of the PEP-based criterion and establishing its asymptotic consistency with the AMI-lower-bound-based criterion in the high-SNR regime. The analysis further reveals how the dominant design factors evolve as the channel transitions from Rayleigh fading to AWGN. Numerical results demonstrate that the proposed framework generates near-optimal codebooks under various channel conditions and achieves significant performance gains over representative benchmarks across a wide range of \(\kappa\) values, particularly in the error-floor region, highlighting its potential for ultra-reliable SCMA applications.

The analysis assumes normalized average received powers across the UEs, corresponding to substantial mitigation of large-scale power differences through uplink power control. Extending the proposed framework to jointly account for residual large-scale fading, power control, and SCMA codebook design constitutes an interesting direction for future work.

	\appendices

	\section{Proofs of the AMI Bounds and High-SNR Analysis}
\label{app:ami}
\noindent\textit{1) Proof of Lemma 1:}
If we define the noiseless received vector associated with $\mathbf{S}$ as $\mathbf{z}(\mathbf{S})=		\sqrt{\frac{E_{\mathrm{s}}}{d_{\mathrm f}}}
\left[
\mathbf{h}_1^{\top}\mathbf{s}_1,\ldots,
\mathbf{h}_K^{\top}\mathbf{s}_K
\right]^{\top}$, with $\mathbf{n}\sim \mathcal{CN}(0,N_0\mathbf{I}_K)$, the received signal in~\eqref{up_Riciank} can be written as $\mathbf{r}=\mathbf{z}(\mathbf{S})+\mathbf{n}$. For any \(\mathbf S\in\mathcal S\), the corresponding conditional PDF of \(\mathbf{r}\) is given by
\begin{equation}
	\label{conditional_pdf}
	p(\mathbf{r}\mid\mathbf{S},\mathbf{H})
	=
	\frac{1}{(\pi N_0)^K}
	\exp\left(
	-\frac{
		\|\mathbf{r}-\mathbf{z}(\mathbf{S})\|^2
	}{N_0}
	\right).
\end{equation}
Consequently, we have
\begin{equation}
	\label{likelihood_ratio}
	\begin{aligned}
		\frac{
			p(\mathbf{r}\mid\hat{\mathbf{S}},\mathbf{H})
		}{
			p(\mathbf{r}\mid\mathbf{S},\mathbf{H})
		}
		&=
		\exp\left[
		-\frac{
			\|\mathbf{r}-\mathbf{z}(\hat{\mathbf{S}})\|^2
			-\|\mathbf{r}-\mathbf{z}(\mathbf{S})\|^2
		}{N_0}
		\right] \\
		&=
		\exp[-d(\mathbf{S},\hat{\mathbf{S}})],
	\end{aligned}
\end{equation}
where
\begin{equation}
	\label{d_pair}
	d(\mathbf{S},\hat{\mathbf{S}})
	=
	\frac{1}{N_0}
	\sum_{k\in\rho(\mathbf{S},\hat{\mathbf{S}})}
	\left[
	\left|
	n_k+\sqrt{\frac{E_s}{d_{\mathrm f}}}\Delta_k
	\right|^2
	-|n_k|^2
	\right]
\end{equation}
and $\Delta_k$ is defined as in \eqref{up_RicianEU}. For fixed \((\mathbf{S},\mathbf{H})\), averaging over \(\mathbf{r}|\mathbf{S},\mathbf{H}\) is equivalent to averaging over \(\mathbf{n}\). Substituting \eqref{likelihood_ratio} and \eqref{d_pair} into \eqref{MI_definition} and averaging over the channel realizations according to \eqref{AMI_exact} then yields
\begin{equation}
	\label{AMI_noise_form}
	\begin{aligned}
		I_{\mathrm{AMI}}
		&=
		L\log_2M
		-\frac{1}{M^L}
		\sum_{\mathbf{S}\in\mathcal{S}}
		\mathrm{E}_{\mathbf{H},\mathbf{n}}
		\left[
		\log_2
		\left(
		\sum_{\hat{\mathbf{S}}\in\mathcal{S}}
		e^{-d(\mathbf{S},\hat{\mathbf{S}})}
		\right)
		\right].
	\end{aligned}
\end{equation}


Since the log-sum-exp function is convex, Jensen's inequality is first applied with respect to \(\mathbf n\) and subsequently with respect to \(\mathbf H\), yielding a lower bound on the expectation term in \eqref{AMI_noise_form} and hence an upper bound on the AMI.

For a fixed pair $(\mathbf S,\hat{\mathbf S})$ and channel realization $\mathbf H$, $\Delta_k$ is deterministic, i.e., independent of the noise realization $\mathbf n$. Since $\mathrm E_{\mathbf n}\left[n_k\right]=0$, the noise-dependent cross term obtained by expanding \eqref{d_pair} vanishes upon averaging over $\mathbf n$. Consequently, we have
\begin{equation}
	\label{noise_avg_distance}
	\mathrm E_{\mathbf n}
	\!\left[
	d(\mathbf S,\hat{\mathbf S})
	\right]
	=
	\lambda\, D(\mathbf S,\hat{\mathbf S}),
\end{equation}
where 
the channel-dependent squared pairwise distance is defined as
\begin{equation}
	\label{D_definition}
	D(\mathbf S,\hat{\mathbf S})
	\triangleq
	\sum_{k\in\rho(\mathbf S,\hat{\mathbf S})}
	|\Delta_k|^2.
\end{equation}

For each fixed channel realization $\mathbf H$, applying Jensen's inequality with respect to $\mathbf n$ to the conditional expectation in \eqref{AMI_noise_form} yields
\begin{equation}
	\label{jensen_step1}
	\mathrm E_{\mathbf n}
	\!\left[
	\log_2
	\sum_{\hat{\mathbf S}\in\mathcal S}
	e^{-d(\mathbf S,\hat{\mathbf S})}
	\,\middle|\,\mathbf H
	\right]
	\geq
	\log_2
	\sum_{\hat{\mathbf S}\in\mathcal S}
	e^{-\lambda D(\mathbf S,\hat{\mathbf S})}.
\end{equation}
Taking the expectation with respect to $\mathbf H$ on both sides of \eqref{jensen_step1} preserves the inequality. By the tower property,	$\mathrm{E}_{\mathbf{H},\mathbf{n}}\left[\cdot \right]=\mathrm{E}_{\mathbf{H}}\left[\mathrm{E}_{\mathbf{n}}\left[\cdot|\mathbf{H}\right]\right]$, the left-hand side becomes the joint expectation in \eqref{AMI_noise_form}, giving
\begin{equation}
	\label{tower_step}
	\mathrm E_{\mathbf H,\mathbf n}
	\!\left[
	\log_2
	\sum_{\hat{\mathbf S}\in\mathcal S}
	e^{-d(\mathbf S,\hat{\mathbf S})}
	\right]
	\geq
	\mathrm E_{\mathbf H}
	\!\left[
	\log_2
	\sum_{\hat{\mathbf S}\in\mathcal S}
	e^{-\lambda D(\mathbf S,\hat{\mathbf S})}
	\right].
\end{equation}
From the Rician channel model in \eqref{up_RicianB}-\eqref{up_RicianMV},
$\mathrm E_{\mathbf H}\!\left[|\Delta_k|^2\right]=|\mu_k|^2+2\sigma_k^2$.
Accordingly, the average pairwise distance is defined as
\begin{equation}
	\label{average_distance}
	\begin{aligned}
		\bar D(\mathbf S,\hat{\mathbf S})
		&\triangleq \mathrm{E}_{\mathbf H}\left[D(\mathbf{S},\hat{\mathbf S})\right]=\sum_{k\in\rho(\mathbf S,\hat{\mathbf S})}\mathrm{E}_{\mathbf H}\left[|\Delta_k|^2\right]\\
		&=\sum_{k\in\rho(\mathbf S,\hat{\mathbf S})}	\left(|\mu_k|^2+2\sigma_k^2\right).
	\end{aligned}
\end{equation}
Applying Jensen's inequality a second time, now with respect to $\mathbf H$, to the right-hand side of \eqref{tower_step} yields
\begin{equation}
	\label{jensen_step2}
	\mathrm E_{\mathbf H}
	\!\left[
	\log_2
	\sum_{\hat{\mathbf S}\in\mathcal S}
	e^{-\lambda D(\mathbf S,\hat{\mathbf S})}
	\right]
	\geq
	\log_2
	\sum_{\hat{\mathbf S}\in\mathcal S}
	e^{-\lambda\bar D(\mathbf S,\hat{\mathbf S})}.
\end{equation}
Combining \eqref{tower_step} and \eqref{jensen_step2}, and
substituting the resulting lower bound into \eqref{AMI_noise_form} yields the following upper bound on the AMI:
\begin{equation}
	\label{AMI_upper_bound}
	I_{\mathrm{AMI}}
	\leq I_{\mathrm{UP}}
	=
	L\log_2M
	-
	\frac{1}{M^L}
	\sum_{\mathbf S\in\mathcal S}
	\log_2
	\left[
	\sum_{\hat{\mathbf S}\in\mathcal S}
	e^{-\lambda\bar D(\mathbf S,\hat{\mathbf S})}
	\right].
\end{equation}
Finally, substituting the definition of \(\bar D(\mathbf S,\hat{\mathbf S})\) in \eqref{average_distance} into \eqref{AMI_upper_bound} yields \eqref{AMI_UP_lamma1}, completing the proof of Lemma 1.
\\
\\
\noindent\textit{2) Proof of Lemma 2:}
To derive a lower bound on the AMI, we first obtain an upper bound on the expectation term in \eqref{AMI_noise_form}. Since $\Delta_k=0$ for $k\notin\rho(\mathbf S,\hat{\mathbf S})$,
\eqref{d_pair} can equivalently be expressed over all $K$ REs as
\begin{equation}
	\label{CD_AMI_lb}
	-d(\mathbf S,\hat{\mathbf S})
	=
	\frac{
		\|\mathbf n\|^2
		-
		\|\mathbf n+\mathbf z(\mathbf S)-\mathbf z(\hat{\mathbf S})\|^2
	}{N_0}.
\end{equation}
For notational convenience, define
$\mathbf{\xi}_{\hat{\mathbf S}}\triangleq\mathbf z(\mathbf S)-\mathbf z(\hat{\mathbf S})$.
Then, we have
\begin{equation}
	\label{factor_exp_step}
	\sum_{\hat{\mathbf S}\in\mathcal S}
	e^{-d(\mathbf S,\hat{\mathbf S})}
	=
	e^{\|\mathbf n\|^2/N_0}
	\sum_{\hat{\mathbf S}\in\mathcal S}
	e^{-\|\mathbf n+\mathbf{\xi}_{\hat{\mathbf S}}\|^2/N_0}.
\end{equation}
Taking $\log_2(\cdot)$ on both sides and averaging over
$\mathbf n$ yields
\begin{equation}
	\label{log_expectation_step}
	\begin{aligned}
		&\mathrm E_{\mathbf n}\!\left[
		\log_2
		\left(
		\sum_{\hat{\mathbf S}\in\mathcal S}
		e^{-d(\mathbf S,\hat{\mathbf S})}
		\right)\right] \\
		&=
		\frac{K}{\ln 2}
		+
		\mathrm E_{\mathbf n}\!\left[
		\log_2
		\left(
		\sum_{\hat{\mathbf S}\in\mathcal S}
		e^{-\|\mathbf n+\mathbf{\xi}_{\hat{\mathbf S}}\|^2/N_0}
		\right)\right],
	\end{aligned}
\end{equation}
where $\mathrm E_{\mathbf n}[\|\mathbf n\|^2]=KN_0$. Since \(\log_2(\cdot)\) is concave, applying Jensen's inequality to the second term gives
\begin{equation}
	\label{jensen_lb_n_full}
	\begin{aligned}
		&\mathrm E_{\mathbf n}\!\left[
		\log_2\!\left(
		\sum_{\hat{\mathbf S}\in\mathcal S}
		e^{-d(\mathbf S,\hat{\mathbf S})}
		\right)\right] \\
		&\overset{(a)}{\leq}
		\frac{K}{\ln 2}
		+
		\log_2\!\left[
		\sum_{\hat{\mathbf S}\in\mathcal S}
		\mathrm E_{\mathbf n}
		\left\{
		e^{-\|\mathbf n+\mathbf{\xi}_{\hat{\mathbf S}}\|^2/N_0}
		\right\}\right] \\
		&\overset{(b)}{=}
		K(\log_2 e-1)
		+
		\log_2\!\left[
		\sum_{\hat{\mathbf S}\in\mathcal S}
		e^{-\frac{\lambda}{2}
			\sum_{k\in\rho(\mathbf S,\hat{\mathbf S})}
			|\Delta_k|^2}
		\right],
	\end{aligned}
\end{equation}
where $(a)$ follows from Jensen's inequality. To obtain $(b)$, we use the Gaussian integral
\begin{equation}
	\label{eq:gaussian_expectation}
	\mathrm E_{\mathbf n}
	\left[
	e^{-\|\mathbf n+\mathbf{\xi}_{\hat{\mathbf S}}\|^2/N_0}
	\right]
	=
	2^{-K}e^{-\|\mathbf{\xi}_{\hat{\mathbf S}}\|^2/(2N_0)},
\end{equation}
which follows by completing the square. In addition,
\begin{equation}
	\label{eq:a_norm}
	\|\mathbf{\xi}_{\hat{\mathbf S}}\|^2
	=
	\frac{E_{\mathrm{s}}}{d_{\mathrm f}}
	\sum_{k\in\rho(\mathbf S,\hat{\mathbf S})}
	|\Delta_k|^2,
\end{equation}
and hence $\|\mathbf{\xi}_{\hat{\mathbf S}}\|^2/(2N_0)=(\lambda/2)\sum_{k\in\rho(\mathbf S,\hat{\mathbf S})}|\Delta_k|^2$. Combining \eqref{eq:gaussian_expectation} and \eqref{eq:a_norm} yields the expression in $(b)$. Using the Rician channel model in \eqref{up_RicianB}-\eqref{up_RicianMV}, we define the per-RE pairwise factor as
\begin{equation}
	\label{Rician_pair_factor}
	\begin{aligned}
		\psi_k
		&\triangleq
		\mathrm E_{\mathbf H}
		\left[
		e^{-\frac{\lambda}{2}|\Delta_k|^2}
		\right] \\
		&=
		\frac{1}{1+\lambda\sigma_k^2}
		\exp\!\left[
		-\frac{\lambda|\mu_k|^2}
		{2(1+\lambda\sigma_k^2)}
		\right].
	\end{aligned}
\end{equation}
Under the assumed independence of the channel coefficients across the involved REs, the expectation of the corresponding pairwise exponential term can be factorized as
\begin{equation}\label{Rician_pair_product}
	\begin{aligned}
		&\mathrm E_{\mathbf H}\!\left[
		\exp\!\left(
		-\frac{\lambda}{2}
		\sum_{k\in\rho(\mathbf S,\hat{\mathbf S})}
		|\Delta_k|^2
		\right)
		\right] \\
		&=\prod_{k\in\rho(\mathbf S,\hat{\mathbf S})}
		\mathrm E_{\mathbf H}\!\left[
		\exp\!\left(
		-\frac{\lambda}{2}|\Delta_k|^2
		\right)
		\right]=
		\prod_{k\in\rho(\mathbf S,\hat{\mathbf S})}\psi_k
		\triangleq
		\Psi(\mathbf S,\hat{\mathbf S}).
	\end{aligned}
\end{equation}
Taking the expectation of \eqref{jensen_lb_n_full} with respect to $\mathbf H$
and applying Jensen's inequality to the concave logarithm yields
\begin{equation}
	\label{jensen_lb_H_exact}
	\begin{aligned}
		&\mathrm E_{\mathbf H,\mathbf n}\!\left[
		\log_2\!\left(
		\sum_{\hat{\mathbf S}\in\mathcal S}
		e^{-d(\mathbf S,\hat{\mathbf S})}
		\right)\right] \\
		&\leq
		K(\log_2e-1)
		+
		\log_2\!\left[
		\sum_{\hat{\mathbf S}\in\mathcal S}
		\Psi(\mathbf S,\hat{\mathbf S})
		\right],
	\end{aligned}
\end{equation}
where \eqref{Rician_pair_product} has been used to evaluate the channel expectation. 
Since the transmitted signal matrices are equiprobable, averaging \eqref{jensen_lb_H_exact} over $\mathbf S\in\mathcal S$ and applying Jensen's inequality once more to the resulting average of the concave logarithm yields
\begin{equation}
	\label{jensen_S_avg}
	\begin{aligned}
		&\frac{1}{M^L}\sum_{\mathbf S\in\mathcal S}\mathrm E_{\mathbf H,\mathbf n}\!\left[
		\log_2\!\left(
		\sum_{\hat{\mathbf S}\in\mathcal S}
		e^{-d(\mathbf S,\hat{\mathbf S})}
		\right)\right] \\
		&\leq
		K(\log_2e-1)+
		\log_2\!\left[
		\frac{1}{M^L}
		\sum_{\mathbf S\in\mathcal S}
		\sum_{\hat{\mathbf S}\in\mathcal S}
		\Psi(\mathbf S,\hat{\mathbf S})
		\right].
	\end{aligned}
\end{equation}
Substituting \eqref{jensen_S_avg} into \eqref{AMI_noise_form} and noting that $-\log_2(M^{-L})=L\log_2M$, the resulting AMI lower bound can be expressed as
\begin{equation}
	\label{AMI_lower_bound}
	\begin{aligned}
		I_{\mathrm{AMI}}
		\geq I_{\mathrm{LB}}
		&=
		2L\log_2M
		-K(\log_2e-1) \\
		&\quad
		-\log_2\!\left[
		\sum_{\mathbf S\in\mathcal S}
		\sum_{\hat{\mathbf S}\in\mathcal S}
		\Psi(\mathbf S,\hat{\mathbf S})
		\right].
	\end{aligned}
\end{equation}
Substituting the definitions of $\psi_k$ and $\Psi(\mathbf S,\hat{\mathbf S})$ in \eqref{Rician_pair_factor} and \eqref{Rician_pair_product} into \eqref{AMI_lower_bound} yields \eqref{AMI_LB_lamma2}, completing the proof of Lemma 2.
\\
\\
\noindent\textit{3) High-SNR Asymptotic Interpretation:}
For a distinct pair $\hat{\mathbf S}\neq\mathbf S$ and sufficiently high SNR such that $\lambda\sigma_k^2\gg1$ for all \(k\in\rho(\mathbf S,\hat{\mathbf S})\), \eqref{Rician_pair_factor} reduces to
\begin{equation}
	\label{AMI_high_SNR_CD}
	\psi_k
	\simeq
	\frac{1}{\lambda\sigma_k^2}
	\exp\!\left(
	-\frac{|\mu_k|^2}{2\sigma_k^2}
	\right).
\end{equation}
Hence, for
$\delta=|\rho(\mathbf S,\hat{\mathbf S})|$, substituting \eqref{AMI_high_SNR_CD} into \eqref{Rician_pair_product} and using the critical-distance definition in \eqref{CD_1} gives
\begin{equation}
	\label{AMI_high_SNR_Psi}
	\Psi(\mathbf S,\hat{\mathbf S})
	\simeq
	\frac{2 \lambda^{-\delta}}
	{d_c^2(\mathbf S,\hat{\mathbf S})}.
\end{equation}
Using the definition of \(\Psi(\mathbf S,\hat{\mathbf S})\) in \eqref{Rician_pair_product}, \eqref{AMI_high_SNR_Psi} is equivalent to \eqref{AMI_high_SNR_lamma} in Section III-D, thereby establishing the high-SNR relation used therein.

To further identify the codebook-dependent contribution to the AMI lower bound, the double sum in \eqref{AMI_lower_bound} can be decomposed into diagonal and off-diagonal terms as
\begin{equation}
	\label{AMI_pair_decomp}
	\sum_{\mathbf S}
	\sum_{\hat{\mathbf S}}
	\Psi(\mathbf S,\hat{\mathbf S})
	=
	M^L+
	\sum_{\mathbf S}
	\sum_{\hat{\mathbf S}\neq\mathbf S}
	\Psi(\mathbf S,\hat{\mathbf S}).
\end{equation}
Since \(\Psi(\mathbf S,\mathbf S)=1\), the diagonal terms contribute the constant $M^L$ and are independent of the codebook design. Hence, only the off-diagonal terms determine the codebook-dependent behavior of the AMI lower bound. From \eqref{AMI_high_SNR_Psi}, each off-diagonal term exhibits the same high-SNR dependence on the diversity order and critical distance as the corresponding PEP derived in Section III-C. Consequently, suppressing the dominant PEP contributions over the complete set of distinct signal-matrix pairs also reduces the codebook-dependent sum in the AMI lower bound. When all pairwise error events are taken into account, the union-bound-based and AMI-lower-bound-based design objectives are asymptotically consistent.

	\section{Representative SCMA Codebooks}\label{app}
	\renewcommand{\thesubsection}{\thesection.\arabic{subsection}}
	\renewcommand{\thesubsubsection}{\thesubsection.\arabic{subsubsection}}

	Two representative codebook designs for 6 UEs, 4 REs and $M=4$ are presented below. In each matrix, the four columns correspond to the codewords labeled $00$, $01$, $10$, and $11$, respectively. 
	Additional designs are available at \url{https://ymchen-nycu.github.io/DNN_Assisted_CBs.xlsx}.
	
	The codebooks listed below represent feasible solutions obtained from the proposed DNN-assisted optimization under the corresponding channel conditions and should not be interpreted as unique solutions, since the numerical optimization may converge to different configurations depending on the initialization and optimization trajectory.
\\
\\
	\noindent\textit{1) Power-Balanced AWGN Codebooks:}
	$$
	\footnotesize
	\mathcal{X}_1\hspace{-0.12cm}:\hspace{-0.12cm}\left[
	\begin{smallmatrix}
	0.3303 - 0.1707i & 0.3296 - 0.1713i & -0.3290 + 0.1686i & -0.3307 + 0.1728i\\
	0.7124 - 0.4524i & 0.4487 + 1.6756i & -0.4494 - 1.6772i & -0.7117 + 0.4528i\\
	0.0000 + 0.0000i &  0.0000 + 0.0000i &  0.0000 + 0.0000i &  0.0000 + 0.0000i\\
	0.0000 + 0.0000i &  0.0000 + 0.0000i &  0.0000 + 0.0000i &  0.0000 + 0.0000i
	\end{smallmatrix}
	\right]
	$$
	$$
	\footnotesize\mathcal{X}_2\hspace{-0.12cm}:\hspace{-0.12cm}\left[
	\begin{smallmatrix}
	-0.1264 + 0.6946i & -1.7136 + 0.2438i &  1.6083 - 0.1949i & 0.2318 - 0.7454i\\
	0.0000 + 0.0000i &  0.0000 + 0.0000i &  0.0000 + 0.0000i &  0.0000 + 0.0000i\\
	0.1207 + 0.3882i & -0.5878 - 0.0839i & 0.1191 + 0.3851i & 0.3462 - 0.6867i\\
	0.0000 + 0.0000i &  0.0000 + 0.0000i &  0.0000 + 0.0000i &  0.0000 + 0.0000i
	\end{smallmatrix}
	\right]
	$$
	$$
	\footnotesize\mathcal{X}_3\hspace{-0.12cm}:\hspace{-0.12cm}\left[
	\begin{smallmatrix}
	0.7696 + 1.5064i & -0.1684 - 0.3293i & -0.7710 - 1.5035i &  0.1681 + 0.3292i\\
	0.0000 + 0.0000i &  0.0000 + 0.0000i &  0.0000 + 0.0000i &  0.0000 + 0.0000i\\
	0.0000 + 0.0000i &  0.0000 + 0.0000i &  0.0000 + 0.0000i &  0.0000 + 0.0000i\\
	0.2249 - 0.4837i & -0.2831 - 0.7996i & -0.2282 + 0.4859i &  0.2842 + 0.7983i
	\end{smallmatrix}
	\right]
	$$
	$$
	\footnotesize\mathcal{X}_4\hspace{-0.12cm}:\hspace{-0.12cm}\left[
	\begin{smallmatrix}
	0.0000 + 0.0000i &  0.0000 + 0.0000i &  0.0000 + 0.0000i &  0.0000 + 0.0000i\\
	1.2915 + 0.1574i & -1.2950 - 0.1578i &  0.3117 + 0.0386i & -0.3086 - 0.0388i\\
	0.0646 - 0.0475i & -0.0671 + 0.0489i & -1.2604 + 0.7789i &  1.2621 - 0.7801i\\
	0.0000 + 0.0000i &  0.0000 + 0.0000i &  0.0000 + 0.0000i &  0.0000 + 0.0000i
	\end{smallmatrix}
	\right]
	$$
	$$
	\footnotesize
	\mathcal{X}_5\hspace{-0.12cm}:\hspace{-0.12cm}\left[
	\begin{smallmatrix}
	0.0000 + 0.0000i &  0.0000 + 0.0000i  & 0.0000 + 0.0000i &  0.0000 + 0.0000i\\
	-0.6453 - 0.0789i & -0.6475 - 0.0793i &  0.6484 + 0.0798i &  0.6460 + 0.0782i\\
	0.0000 + 0.0000i &  0.0000 + 0.0000i  & 0.0000 + 0.0000i &  0.0000 + 0.0000i\\
	-1.1135 + 0.5798i & -0.0022 - 1.2539i &  0.0019 + 1.2555i &  1.1134 - 0.5806i
	\end{smallmatrix}
	\right]
	$$
	$$
	\footnotesize
	\mathcal{X}_6\hspace{-0.12cm}:\hspace{-0.12cm}\left[
	\begin{smallmatrix}
	0.0000 + 0.0000i &  0.0000 + 0.0000i &  0.0000 + 0.0000i &  0.0000 + 0.0000i\\
	0.0000 + 0.0000i &  0.0000 + 0.0000i &  0.0000 + 0.0000i &  0.0000 + 0.0000i\\
	0.0664 - 0.0463i & -0.7638 - 1.3058i & -0.0656 + 0.0460i &  0.7640 + 1.3066i\\
	-1.1156 - 0.6742i & -0.0395 - 0.0240i &  1.1173 + 0.6757i &  0.0360 + 0.0210i
	\end{smallmatrix}
	\right]
	$$
\\
\noindent\textit{2) Rician Codebooks for $\kappa=8$:}	
$$
\footnotesize
\mathcal{X}_1\hspace{-0.12cm}:\hspace{-0.12cm}\left[
\begin{smallmatrix}
-0.0882 + 0.3823i & 0.0926 - 0.4378i & -0.2409 + 1.1349i & 0.2366 - 1.0794i\\
0.4530 - 1.5281i & 0.1100 - 0.3843i & -0.1479 + 0.5347i & -0.4150 + 1.3778i\\
0.0000 + 0.0000i &  0.0000 + 0.0000i &  0.0000 + 0.0000i &  0.0000 + 0.0000i\\
0.0000 + 0.0000i &  0.0000 + 0.0000i &  0.0000 + 0.0000i &  0.0000 + 0.0000i
\end{smallmatrix}
\right]
$$
$$
\footnotesize\mathcal{X}_2\hspace{-0.12cm}:\hspace{-0.12cm}\left[
\begin{smallmatrix}
0.1028 - 0.2230i & 0.3899 - 1.3381i & -0.1029 + 0.2230i & -0.3897 + 1.3381i\\
0.0000 + 0.0000i &  0.0000 + 0.0000i &  0.0000 + 0.0000i &  0.0000 + 0.0000i\\
1.3204 + 0.4407i & 0.2407 + 0.0429i & -1.3202 - 0.4406i & -0.2411 - 0.0429i\\
0.0000 + 0.0000i &  0.0000 + 0.0000i &  0.0000 + 0.0000i &  0.0000 + 0.0000i
\end{smallmatrix}
\right]
$$
$$
\footnotesize\mathcal{X}_3\hspace{-0.12cm}:\hspace{-0.12cm}\left[
\begin{smallmatrix}
1.3082 + 0.3263i & 0.3394 + 0.0851i & -1.3082 - 0.3263i & -0.3395 - 0.0850i\\
0.0000 + 0.0000i &  0.0000 + 0.0000i &  0.0000 + 0.0000i &  0.0000 + 0.0000i\\
0.0000 + 0.0000i &  0.0000 + 0.0000i &  0.0000 + 0.0000i &  0.0000 + 0.0000i\\
-0.2268 + 0.0974i & 1.2936 - 0.5704i & 0.2267 - 0.0975i & -1.2937 + 0.5703i
\end{smallmatrix}
\right]
$$
$$
\footnotesize\mathcal{X}_4\hspace{-0.12cm}:\hspace{-0.12cm}\left[
\begin{smallmatrix}
0.0000 + 0.0000i &  0.0000 + 0.0000i &  0.0000 + 0.0000i &  0.0000 + 0.0000i\\
-0.6226 - 0.1263i & 0.6308 + 0.1330i & -1.1413 - 0.2809i &  1.1330 + 0.2741i\\
0.4883 - 1.3991i & 0.1304 - 0.3690i & -0.1749 + 0.4958i & -0.4437 + 1.2721i\\
0.0000 + 0.0000i &  0.0000 + 0.0000i &  0.0000 + 0.0000i &  0.0000 + 0.0000i
\end{smallmatrix}
\right]
$$
$$
\footnotesize
\mathcal{X}_5\hspace{-0.12cm}:\hspace{-0.12cm}\left[
\begin{smallmatrix}
0.0000 + 0.0000i &  0.0000 + 0.0000i &  0.0000 + 0.0000i &  0.0000 + 0.0000i\\
1.3869 + 0.4777i & 0.4080 + 0.1276i & -1.2719 - 0.4210i & -0.5231 - 0.1844i\\
0.0000 + 0.0000i &  0.0000 + 0.0000i  & 0.0000 + 0.0000i &  0.0000 + 0.0000i\\
0.4619 - 0.2145i & -0.3458 + 0.1558i & 1.0845 - 0.4903i & -1.2008 + 0.5491i
\end{smallmatrix}
\right]
$$
$$
\footnotesize
\mathcal{X}_6\hspace{-0.12cm}:\hspace{-0.12cm}\left[
\begin{smallmatrix}
0.0000 + 0.0000i &  0.0000 + 0.0000i &  0.0000 + 0.0000i &  0.0000 + 0.0000i\\
0.0000 + 0.0000i &  0.0000 + 0.0000i &  0.0000 + 0.0000i &  0.0000 + 0.0000i\\
-1.0422 - 0.3802i & 0.3735 + 0.1359i &  1.0924 + 0.3934i & -0.4236 - 0.1490i\\
0.5890 + 1.3117i & -0.6499 - 1.4520i & 0.2245 + 0.5041i & -0.1635 - 0.3637i
\end{smallmatrix}
\right]
$$

\end{document}